\documentclass[12pt,en,authoryear]{elegantpaper}

\usepackage{caption}
\usepackage{subcaption}
\usepackage{wrapfig}
\usepackage{nicefrac}
\usepackage{mathtools}
\mathtoolsset{showonlyrefs=true}
\usepackage{relsize}
\usepackage{thmtools}
\usepackage{thm-restate}
\newtheorem{claim}[theorem]{Claim}

\newcommand{\bsucc}{\boldsymbol{\succ}}
\newcommand{\bbR}{\mathbb{R}}

\newcommand{\calR}{\mathcal{R}}
\newcommand{\cR}{\mathcal{R}}
\newcommand{\cF}{\mathcal{F}}

\newcommand{\pairwise}{pairwise exchange }
\title{Strategy-proof and Envy-free Mechanisms for House Allocation\footnote{We thank Haluk Ergin, David Ahn, Yuichiro Kamada, Chris Shannon and Shachar Kariv for their helpful comments and suggestions.}}
\author{Priyanka Shende \footnote{Email: priyanka.s@berkeley.edu} \\ University of California, Berkeley \and Manish Purohit\footnote{Email: mpurohit@google.com} \\ Google Research, Mountain View}

\begin{document}

\maketitle

\begin{abstract}
We consider the problem of allocating indivisible objects to agents when agents have strict preferences over objects. There are inherent trade-offs between competing notions of efficiency, fairness and incentives in assignment mechanisms. It is, therefore, natural to consider mechanisms that satisfy two of these three properties in their strongest notions, while trying to improve on the third dimension. In this paper, we are motivated by the following question: Is there a strategy-proof and envy-free random assignment mechanism more efficient than equal division?

Our contributions in this paper are twofold. First, we further explore the incompatibility between efficiency and envy-freeness in the class of strategy-proof mechanisms. We define a new notion of efficiency that is weaker than ex-post efficiency and prove that any strategy-proof and envy-free mechanism must sacrifice efficiency even in this very weak sense. Next, we introduce a new family of mechanisms called Pairwise Exchange mechanisms and make the surprising observation that strategy-proofness is equivalent to envy-freeness within this class.
We characterize the set of all neutral and strategy-proof (and hence, also envy-free) mechanisms in this family and show that they admit a very simple linear representation.
\end{abstract}

\newpage

% Paper body
\section{Introduction}
\label{sec:introduction}

The house allocation problem is a fundamental resource allocation problem that deals with the assignment of indivisible objects to agents without the use of monetary transfers. Monetary transfers are undesirable %(and sometimes even illegal!)
in many real-world applications such as placement of students to public schools \citep{abdulkadirouglu2003school}, course allocation \citep{budish2011combinatorial}, organ donation \citep{roth2005pairwise}, and on-campus housing allocation \citep{chen2002improving}.
In the classical \emph{ordinal} setting, each agent reports a strict preference ranking over the set of alternatives and the mechanism outputs an assignment of objects to agents. An ideal mechanism used in such markets must have certain qualities; it must be non-manipulable, efficient and fair.

Since objects are indivisible, any assignment of objects to agents is bound to be perceived as unfair, ex-post. Randomization is, therefore, commonly used as a tool to restore fairness from an ex-ante perspective. Perhaps the most natural randomized mechanism is the \emph{Random Serial Dictatorship} (RSD) mechanism, also known as \emph{Random Priority} mechanism. In this mechanism, agents are ordered uniformly at random, and each agent successively chooses her favorite object from the set of available objects according to that order. RSD is known to satisfy a number of attractive properties \citep{abdulkadirouglu1998random,bogomolnaia2001new}. It is strategy-proof, meaning that revealing true preferences is always a dominant strategy for every agent. It is also ex-post efficient, which implies that it always induces efficient eventual outcomes.
However, it satisfies fairness only in the weak sense of equal treatment of equals (where agents with identical preferences face identical lotteries over objects).

A stronger notion of fairness is envy-freeness. Introduced by \cite{foley1967resource}, the classic definition of envy-freeness requires that each agent should prefer her allocation to anyone else's allocation. This notion is often considered as the gold standard of fairness in many different settings such as resource allocation~\citep{foley1967resource}, cake-cutting~\citep{robertson1998cake}, and rent division~\citep{edward1999rental}. In the context of random assignment mechanisms, \cite{bogomolnaia2001new} formulate this property using the first-order stochastic dominance relation. 
They proposed the \emph{Probablistic Serial} (PS) mechanism that satisfied stronger efficiency and fairness properties than RSD. PS is ordinally efficient\footnote{A mechanism is ordinally efficient if its outcome is not first-order stochastically dominated by any other random assignment.} and envy-free. However, this mechanism is not strategy-proof. 

Unfortunately, a growing body of work~\citep{bogomolnaia2001new,nesterov2017fairness,zhou1990conjecture,martini2016strategy} has demonstrated the inherent incompatibility between efficiency, fairness, and strategy-proofness.
\citet{bogomolnaia2001new} proved that no mechanism  simultaneously satisfies strategy-proofness, ordinal efficiency and equal treatment of equals. More recently, \citet{nesterov2017fairness} showed the incompatibility between strategy-proofness, envy-freeness and ex-post efficiency.
Given these trade-offs, it is natural to consider mechanisms that can be designed when one would like to have two of the three properties of efficiency, fairness and strategy-proofness to be satisfied in their strongest notions, where the choice of the properties depends on the application, while trying to improve on the third dimension. 

Since the primary motivation for randomization in house allocation mechanisms is to provide fairness guarantees, in this paper, we focus our attention on strategy-proof random allocation mechanisms that satisfy envy-freeness. Indeed, there has been a resurgence of interest in fairness at the intersection of Economics and Computer Science in recent years. 
See, for instance, the recent survey by \citet{moulin2019fair} and the EC workshop\footnote{Workshop on Fairness at ACM Conference on Economics and Computation, 2019: \url{https://users.cs.duke.edu/~rupert/fair-division-ec19/index.html}} on fair resource allocation for an excellent overview. 
To the best of our knowledge, the equal division (ED) mechanism that allocates each object equally among all agents is the only known mechanism that is strategy-proof and envy-free. However, since this mechanism completely ignores agents' preferences, it is almost always inefficient. This raises the natural question: are there other strategy-proof and envy-free mechanisms that are more efficient than the equal division mechanism?

\subsection{Our Contributions}

We first show a strong impossibility result to demonstrate that strategy-proof and envy-free mechanisms must sacrifice efficiency even in a very weak sense. We define a notion of \emph{contention-free efficiency} that is much weaker than ex-post efficiency and show that no strategy-proof mechanism that satisfies envy-freeness can be contention-free efficient. Our result thus strengthens and subsumes the hardness result by \citet{bogomolnaia2001new} and \citet{nesterov2017fairness} regarding the incompatibility of ex-post efficiency, strategy-proofness and envy-freeness.

In order to design new strategy-proof and envy-free mechanisms, we first argue, in Section \ref{subsec:random_mechanisms_representation}, that any random mechanism can be thought of as starting from equal division mechanism and allowing every pair of agents to exchange probabilistic shares of objects. These exchanges can be represented using preference profile specific \emph{transfer functions}. Next, we restrict our attention to a simple class of transfer functions that depend only on the preferences of the pair of agents involved in an exchange and characterize the family of mechanisms called \emph{\pairwise mechanisms} that admit such transfers. In Section \ref{subsec:equiv}, we demonstrate a surprising equivalence result between strategy-proofness and envy-freeness within this class of mechanisms. We provide a characterization for the set of neutral, strategy-proof and envy-free mechanisms in this family in Section \ref{subsection:sp_neutral_characterize}. These mechanisms have a simple linear representation, where the transfer for any object within a pair of agents is determined by the relative rank for that object in the agents' preference orderings. Finally in Section \ref{sec:pareto-dominance}, we characterize the set of all Pareto-efficient mechanisms within this class.

\subsection{Other Related Work}
This paper adds to the broad literature on random assignment of indivisible objects that was pioneered by \citet{hylland1979efficient}. They adapt the competitive equilibrium from equal incomes (CEEI) solution to define a pseudo-market mechanism that elicits agents' von Neumann-Morgenstern preferences over individual objects and gives a solution that is efficient and fair with respect to the utility functions (i.e., ex-ante efficient and envy-free). However, the mechanism is not strategy-proof. \citet{zhou1990conjecture}, in fact, show (proving a conjecture by \citet{gale1987}) that there exists no strategy-proof mechanism that satisfies ex-ante efficiency and equal treatment of equals.

When agents' report only ordinal preferences over individual objects, the simplest and the most widely known strategy-proof mechanism is serial dictatorship (SD) \citep{svensson1994queue,satterthwaite1981strategy}: using a fixed ordering of agents, each agent successively chooses her most preferred object from the set of remaining objects. \citet{abdulkadirouglu1998random} show that these are the only Pareto efficient matching mechanisms. However, while these are very unfair, using a random ordering of agents results in restoration of fairness in the sense of equal treatment of equals. The resulting mechanism, RSD, was also analyzed by \citet{abdulkadirouglu1998random} and shown to be equivalent to the \emph{core from random endowments} mechanism, where each agent is initially endowed with an object that is chosen uniformly at random and the mechanism then uses Gale's \emph{Top Trading Cycles} (TTC) algorithm \citep{shapley1974cores} to arrive at a random assignment. 

In another successful line of research, several papers have focused on characterizing (families of) mechanisms that satisfy certain desirable properties. For deterministic mechanisms, \citet{svensson1999strategy} proves that serial dictatorships are the only group strategy-proof and neutral mechanisms. \citet{papai2000strategyproof} introduces a class of mechanisms called hierarchical
exchange mechanisms and proves that these mechanisms characterize the class of group strategy-proof, Pareto-efficient and reallocation-proof mechanisms. More recently, \citet{pycia2017incentive} have shown that their trading-cycles mechanisms characterize the full class of group-strategy proof and Pareto-efficient mechanisms. Within random mechanisms, when there are three agents and three objects \cite{bogomolnaia2001new} show that RSD is the unique mechanism that satisfies strategy-proofness, ex-post efficiency and equal treatment of equals while PS is the unique mechanism that satisfies ordinal efficiency, evny-freeness and weak strategy-proofness. \cite{bogomolnaia2012probabilistic,heo2014extended,heo2014probabilistic,hashimoto2014two,heo2015characterization} provide other axiomatic characterizations of PS. In large regular markets, \citet{liu2016ordinal} show that there is a unique mechanism that is ordinally efficient, envy-free and strategy-proof and all uniform randomizations over known deterministic mechanisms such as serial dictatorships,  hierarchical exchange and trading-cycles mechanisms coincide with this unique mechanism. \cite{chambers2004consistency} introduces a notion of probabilistic consistency and showed that ED is the only mechanism that satisfies probabilistic consistency and equal treatment of equals.
\section{Preliminaries}
\label{sec:prelims}

\subsection{Model and Notation}
\label{sec:model}

In this section, we formally define the canonical house allocation problem (also called in the literature as the assignment problem), random assignment mechanisms and properties of these mechanisms. 

Let $N$ be the set of agents and $O$ be the set of objects. Throughout this paper, we assume that the sets $N$ and $O$ are fixed and finite with $|N|=|O|=n \geq 3$. Each agent $i \in N$ has a strict preference relation $\succ_i$ on $O$. The corresponding weak preference relation on $O$ is denoted by $\succeq_i$. 
When $\succ= a_1 \succ a_2 \succ \ldots \succ a_n$, for brevity, we will denote such a preference relation by $\succ = \langle a_1, a_2, \ldots, a_n\rangle$.
A set of individual preferences of all agents constitutes a preference profile $\boldsymbol{\succ}=(\succ_i)_{i \in N}$. Let $\boldsymbol{\succ}_{-i} \ =\ \boldsymbol{\succ} \setminus \{\succ_i\}$ denote the set of preferences of all agents other than agent $i$. We will use $\bsucc=(\succ_i=\succ,\succ_j=\succ',\bsucc_{-\{i,j\}})$ to denote a preference profile when agent $i$'s preference is $\succ$, agent $j$'s preference is $\succ'$ and all other agent's preferences are given by $\bsucc_{-\{i,j\}}$. Let $\mathcal{R}$ be the set of all individual preferences and $\mathcal{R}^n$ be the set of all possible preference profiles. 

A \emph{deterministic} assignment is a bijection from $N$ to $O$, where every agent receives one object and every object is assigned to exactly one agent. Let $\mathcal{D}$ be the set of all deterministic assignments.
A \emph{random} assignment is a probability distribution over deterministic assignments. A random assignment $P=\Big[P_{i,a}\Big]_{i \in N, a \in O}$ can be represented as a doubly stochastic matrix of size $n \times n$, where each element $P_{i,a}$ of the matrix $P$ represents the probability with which agent $i$ is assigned object $a$\footnote{For convenience in this paper, we will often think of each object as an infinitely divisible good of one unit that will be distributed among $n$ agents. Allocating a fractional unit $x$ of an object $a$ to agent $i$ is interpreted as setting $P_{i,a} = x$.}. 
Let $\mathcal{L}(\mathcal{D})$ be a set of all possible random assignments. The $i^\text{th}$ row of the matrix $P$, $P_i$, represents the allocation received by agent $i$ in the random assignment. This allocation is simply a probability distribution over the set of objects $O$. The set of all random allocations will be denoted by $\mathcal{L}(O)$. 

In order to compare random assignments, we extend agents' preferences over the set of objects $O$ to the set of random allocations $\mathcal{L}(O)$. Let $U(\succ,a) = \{o \in O \mid o \succeq a\}$ be the set of objects that are weakly preferred to object $a$ according to the preference relation $\succ$. For example, if $\succ = \langle a_1, a_2, \ldots, a_n\rangle$, then $U(\succ, a_k) = \{a_1, \ldots, a_k\}$. Agent $i$ prefers a random allocation $P_i \in \mathcal{L}(O)$ to another random allocation $Q_i \in \mathcal{L}(O)$ (denoted as $P_i \geq_i Q_i$) if and only if the random allocation $P_i$ first-order stochastically dominates $Q_i$ according to agent $i$'s preference $\succ_i$. Formally,
\begin{align*}
    P_i \geq_i Q_i \Longleftrightarrow \sum_{o \in U(\succ_i, a)} P_{i,o} \geq \sum_{o \in U(\succ_i, a)} Q_{i,o} \ , \forall \ a \in O
\end{align*}
If $P_i \geq_i Q_i$, and in addition, there is an object $b \in O$, such that $\sum_{o \in U(\succ_i, b)} P_{i,o} > \sum_{o \in U(\succ_i, b)} Q_{i,o}$, then we say that agent $i$ strictly prefers $P_i$ to $Q_i$, which is denoted by $P_i >_i Q_i$. Lastly, a random assignment $P$ \emph{dominates} another assignment $Q$ if every agent prefers the random allocation that she receives in $P$ to her random allocation in $Q$. That is, $P$ \emph{dominates} $Q$, if $\forall \, i \in N$, $P_i \geq_i Q_i$.

A random assignment \emph{mechanism} is a mapping, $ \varphi:\mathcal{R}^n \rightarrow \mathcal{L}(\mathcal{D})$, that associates each preference profile $\bsucc \in \mathcal{R}^n$ with some random assignment $P \in \mathcal{L}(\mathcal{D})$.
For ease of exposition, we often use $P^{(x)}$:=$\varphi(\bsucc^{(x)})$ to denote the random assignment associated with preference profile $\bsucc^{(x)}$ in the mechanism $\varphi$.

We define some additional notation that we will use throughout the paper. Let $[n] = \{1, 2, \ldots, n\}$. Let $\pi:O \rightarrow O$ be a permutation, i.e. a bijection from $O$ to itself. For any preference relation $\succ$ and permutation $\pi$, let $\pi(\succ)$ denote the preference relation obtained by applying the permutation $\pi$ to every object in $\succ$ in order. In other words, objects are re-labeled according to the function $\pi$. Formally, if $\succ = \langle a_1, a_2, \ldots, a_n\rangle$ and $\pi$ is a permutation, then $\pi(\succ) = \langle \pi(a_{1}), \pi(a_{2}), \ldots, \pi(a_{n})\rangle$. Similarly for any preference profile $\bsucc = (\succ_i)_{i \in N}\in \cR^n$, let $\pi(\bsucc) = (\pi(\succ_i))_{i \in N}$ be the preference profile obtained by re-labeling the objects according to $\pi$. 
Let $\sigma(\succ, k)$ denote the $k^{\text{th}}$ most preferred object according to the preference relation $\succ$. Finally, we use $rank(\succ, a)$ to denote the rank of object $a$ in $\succ$.

\subsection{Properties of Mechanisms}

We now formally define the different notions of efficiency, fairness, and incentive-compatibility that we address in this paper.

\paragraph{Efficiency.} A random mechanism $\varphi$ is said to be \emph{ex-post efficient} if for any preference profile $\bsucc$, the random assignment $\varphi(\bsucc)$ is a distribution over Pareto-optimal deterministic assignments. 
An even stronger notion of efficiency was proposed by \citet{bogomolnaia2001new}. 
 A mechanism $\varphi$ is said to be \emph{ordinally efficient} if for any profile $\bsucc$, the assignment $\varphi(\bsucc)$ is not dominated by any other random assignment.

\paragraph{Incentives.} A mechanism $\varphi$ is said to be \emph{strategy-proof} (SP) if reporting true preferences is a dominant strategy for every agent. Formally, a mechanism $\varphi$ is strategy-proof if for every $i \in N$, for every $\boldsymbol{\succ} \in \mathcal{R}^n$ and $\succ'_i \in \mathcal{R}$, $P^{\boldsymbol{\succ}}_i \geq_i P^{(\succ'_i,\boldsymbol{\succ}_{-i})}_i$ where $P^{\boldsymbol{\succ}}=\varphi(\boldsymbol{\succ})$ and $P^{(\succ'_i,\boldsymbol{\succ}_{-i})}=\varphi(\succ'_i,\boldsymbol{\succ}_{-i})$. 

As shown by \cite{mennle2014partial,mennle2014axiomatic}, strategy-proofness is equivalent to the following three axioms: \emph{swap monotonicity, upper-invariance}, and \emph{lower-invariance}. We define these axioms below.
For any preference relation $\succ_i$, we define the \emph{neighborhood} of $\succ_i$, $\Gamma(\succ_i)$, to be set of all preferences that can arise by swapping any two consecutively ranked objects in the preference $\succ_i$. For example, suppose $O=\{a,b,c,d\}$ and $\succ_i := \langle a, b, c, d \rangle$. Then $\Gamma(\succ_i)=\{\langle b, a, c, d\rangle, \langle a, c, b ,d\rangle, \langle a, b, d, c \rangle \}$. 

A mechanism $\varphi$ is swap-monotonic, if for every agent $i \in N$, for every preference profile $\boldsymbol{\succ} = (\succ_i, \boldsymbol{\succ}_{-i}) \in \mathcal{R}^n$ and every $\succ'_i \in \Gamma(\succ_i)$ with $a \succ_i b$ but $b \succ'_i a$ for some $a,b \in O$, then either $P'_i = P_i$ or $P'_{i,b} > P_{i,b}$ where $P = \varphi(\boldsymbol{\succ})$ and $P' = \varphi(\succ'_i,\boldsymbol{\succ_{-i}})$.
In other words, swap monotonicity requires the mechanism $\varphi$ to either disregard agent $i$'s mis-report or to react to her mis-report by allocating her a higher probability of the object that's been brought up in the preference. 

A mechanism is upper-invariant if no agent, by swapping some object $a$ with a less preferred  $b$ in her preference, can get more allocation for any of the objects that are strictly preferred to object $a$. Formally, $\varphi$ is upper-invariant if for every agent $i \in N$, for every $\boldsymbol{\succ} = (\succ_i, \boldsymbol{\succ}_{-i}) \in \mathcal{R}^n$ and every $\succ'_i \in \Gamma(\succ_i)$ with $a \succ_i b$ but $b \succ'_i a$ for some $a,b \in O$, we have  $P_{i,o}=P_{i,o}' \ \forall o \in U(\succ_i, a) \setminus \{a\}$, where $P = \varphi(\boldsymbol{\succ})$ and $P' = \varphi(\succ'_i,\boldsymbol{\succ_{-i}})$. 

Similarly, a mechanism is lower-invariant  if for every agent $i \in N$, for every $\boldsymbol{\succ} = (\succ_i, \boldsymbol{\succ}_{-i}) \in \mathcal{R}^n$ and every $\succ'_i \in \Gamma(\succ_i)$ with $a \succ_i b$ but $b \succ'_i a$ for some $a,b \in O$, we have  $P_{i,o}=P_{i,o}' \ \forall o \notin U(\succ_i, b)$, where $P = \varphi(\boldsymbol{\succ})$ and $P' = \varphi(\succ'_i,\boldsymbol{\succ_{-i}})$.

\begin{lemma}[\cite{mennle2014partial,mennle2014axiomatic}]\label{lemma:sp_swap_upper_lower}
A mechanism $\varphi$ is strategy proof if and only if it is swap-monotonic, upper invariant, and lower invariant.
\end{lemma}

\paragraph{Fairness.} Different notions of fairness have been considered in literature. \emph{Equal treatment of equals} is a weak fairness criterion that requires two agents with the same reported preferences to get the same random allocations. On the other hand, \emph{envy-freeness} is a well established strong fairness criterion that requires that no agent envies the allocation of any other agent. Formally, for any preference profile $\bsucc = (\succ_k)_{k \in N}$, a random assignment $P$ is envy-free if we have for all $i, j \in N$, $P_i \geq_i P_j$. A mechanism $\varphi$ is envy-free if it always produces envy-free assignments. It can be readily seen that envy-freeness implies equal treatment of equals.

\paragraph{Neutrality.} In this paper, we often restrict our attention to neutral mechanisms. Neutrality is a natural notion of symmetry that restricts the mechanism to treat all objects identically, i.e., the mechanism is invariant to any renaming of objects. Formally, a mechanism $\varphi$ is \emph{neutral} if for any preference profile $\bsucc \in \cR^n$ and permutation $\pi : O \rightarrow O$, if $P = \varphi(\bsucc)$ and $P^\pi = \varphi(\pi(\bsucc))$, then for all agents $i \in N$ and any $a \in O$, we have $P_{i,a} = P^\pi_{i, \pi(a)}$.

\paragraph{Anonymity.} Anonymity is another restriction we place on the mechanisms that we study in this paper. It is a common requirement imposed on a mechanism that requires
them to treat all agents identically, i.e., the mechanism is invariant to any renaming of agents. Let $\pi : N \rightarrow N$ be a permutation of agents. If $\bsucc=(\succ_1,\succ_2,\ldots,\succ_n)$, then $\pi(\bsucc)=(\succ_{\pi(1)},\succ_{\pi(2)},\ldots,\succ_{\pi(n)})$ is the resulting preference profile where agent $i$ now reports the preference that was reported by agent $\pi(i)$ in $\bsucc$, i.e., $\succ_{\pi(i)}$. A mechanism $\varphi$ is \emph{anonymous} if for any preference profile $\bsucc \in \cR^n$ and permutation $\pi : N \rightarrow N$, if $P = \varphi(\bsucc)$ and $P^\pi = \varphi(\pi(\bsucc))$, then for all agents $i \in N$ and any $a \in O$, we have $P_{\pi(i),a} = P^\pi_{i, a}$.

\section{An Impossibility Result}
\label{sec:hardness}

In this section, we present our first main result regarding the inherent trade-offs between the competing notions of efficiency, fairness, and strategy-proofness. Our goal is to demonstrate that any strategy-proof and envy-free mechanism must sacrifice efficiency in even the weakest sense. 

Before stating our theorem, we first introduce a very weak notion of efficiency, which we call \emph{contention-free efficiency}. Intuitively, this notion captures the desideratum that if there is no competition for objects among the different agents, then an efficient mechanism must allocate to each agent her most preferred object\footnote{This property is analogous to the unanimity axiom defined in the social choice literature \citep{muller1977equivalence}.}. We capture this intuition formally as follows.

\begin{definition}[Contention-Free  Profile] A preference profile $\bsucc$ is called a \emph{contention-free} preference profile if every agent prefers a distinct object as her top choice. Formally, a preference profile $\bsucc \in \calR^n$ is \emph{contention-free} if and only if for all $i, j \in N$, $\sigma(\succ_i,1)=\sigma(\succ_j,1) \Leftrightarrow i=j$.
\end{definition}

Let $\mathcal{C} \subset \mathcal{R}^n$ be the set of all contention-free preference profiles.

\begin{definition}[Contention-Free Efficiency] A mechanism $\varphi$ is defined to be \emph{contention-free efficient} if and only if for every contention-free preference profile it allocates to every agent, her most preferred object fully. That is, $\varphi$ is contention-free efficient $\Leftrightarrow \forall \bsucc \in \mathcal{C} \text{ and } \forall i \in N, P_{i,\sigma(\succ_i,1)}=1$, where $P = \varphi(\bsucc)$. 
\end{definition}

Such a notion clearly imposes a very minimal efficiency requirement on a mechanism. Indeed, it places no restrictions at all at any profile that is not contention-free. Further, for any contention-free profile, the unique Pareto-optimal deterministic assignment is one that allocates to every agent her most preferred object. So any ex-post efficient mechanism must also be contention-free efficient. 
Surprisingly, we show that strategy-proofness, envy-freeness and contention-free efficiency are incompatible. 

\begin{theorem}
\label{thm:main_hardness}
For any $n >= 3$, no strategy-proof and envy-free mechanism can be contention-free efficient.
\end{theorem}

We first prove the claim for $n=3$ in the following lemma. Theorem \ref{thm:main_hardness} follows from a reduction to this case.

\begin{lemma}
\label{theorem:impossibility_=3}
For $n = 3$, no strategy-proof and envy-free mechanism can be contention-free efficient.
\end{lemma}

\begin{proof}
Let $N=\{1,2,3\}$ and $O=\{a,b,c\}$ denote the set of agents and objects respectively. Suppose for contradiction that there exists a mechanism $\varphi$ that is strategy-proof, envy-free and contention-free efficient. Recall that, we adopt the notation $P^{(x)} = \varphi(\bsucc^{(x)})$ for any profile $\bsucc^{(x)}$.  We will proceed by considering six preference profiles, which are shown in Table \ref{tab:table_profiles}. The corresponding random assignments given by the mechanism $\varphi$ are shown in Table \ref{tab:table_assignment}. 

\begin{table}[h]
    \begin{subtable}[h]{0.3\textwidth}
        \centering
        \begin{tabular}{llll}
        \\
        \toprule
1 & a & b & c \\
2 & b & a & c \\
3 & c & a & b \\
        \bottomrule
       \end{tabular}
       \caption{Profile A}
       \label{tab:profile_a}
    \end{subtable}
    \begin{subtable}[h]{0.3\textwidth}
        \centering
        \begin{tabular}{llll}
        \\
        \toprule
1 & a & b & c \\
2 & a & b & c \\
3 & c & a & b \\
        \bottomrule
       \end{tabular}
       \caption{Profile B}
       \label{tab:profile_b}
    \end{subtable}
    \begin{subtable}[h]{0.3\textwidth}
        \centering
        \begin{tabular}{llll}
        \\
        \toprule
1 & a & b & c \\
2 & a & c & b \\
3 & c & a & b \\
        \bottomrule
       \end{tabular}
       \caption{Profile C}
       \label{tab:profile_c}
    \end{subtable}\\
    \begin{subtable}[h]{0.3\textwidth}
        \centering
        \begin{tabular}{llll}
        \\
        \toprule
1 & b & a & c \\
2 & a & c & b \\
3 & c & a & b \\
        \bottomrule
       \end{tabular}
       \caption{Profile D}
       \label{tab:profile_d}
    \end{subtable}
    \begin{subtable}[h]{0.3\textwidth}
        \centering
        \begin{tabular}{llll}
        \\
        \toprule
1 & a & b & c \\
2 & a & c & b \\
3 & a & c & b \\
        \bottomrule
       \end{tabular}
       \caption{Profile E}
       \label{tab:profile_e}
    \end{subtable}
    \begin{subtable}[h]{0.3\textwidth}
        \centering
        \begin{tabular}{llll}
        \\
        \toprule
1 & b & a & c \\
2 & a & c & b \\
3 & a & c & b \\
        \bottomrule
       \end{tabular}
       \caption{Profile F}
       \label{tab:profile_f}
    \end{subtable}
     \caption{Six preference profiles to demonstrate incompatibility of strategy-proofness, envy-freeness, and contention-free efficiency.}
     \label{tab:table_profiles}
\end{table}

\begin{table}[h]
    \begin{subtable}[h]{0.3\textwidth}
        \centering
        \begin{tabular}{llll}
 & a & b & c \\
\bottomrule
1 & 1 & 0 & 0 \\
2 & 0 & 1 & 0 \\
3 & 0 & 0 & 1 \\
        \bottomrule
       \end{tabular}
       \caption{Profile A}
       \label{tab:alloc_a}
    \end{subtable}
    \begin{subtable}[h]{0.3\textwidth}
        \centering
        \begin{tabular}{llll}
 & a & b & c \\
\bottomrule
1 & $\nicefrac{1}{2}$ & $\nicefrac{1}{2}$ & 0 \\
2 & $\nicefrac{1}{2}$ & $\nicefrac{1}{2}$ & 0 \\
3 & 0 & 0 & 1 \\
        \bottomrule
       \end{tabular}
       \caption{Profile B}
       \label{tab:alloc_b}
    \end{subtable}
    \begin{subtable}[h]{0.3\textwidth}
        \centering
        \begin{tabular}{llll}
 & a & b & c \\
\bottomrule
1 & $\nicefrac{1}{2}$ & $1-2y$ & $2y-\nicefrac{1}{2}$ \\
2 & $\nicefrac{1}{2}$ & $y$ & $\nicefrac{1}{2}-y$ \\
3 & $0$ & $y$ & $1-y$ \\
        \bottomrule
       \end{tabular}
       \caption{Profile C}
       \label{tab:alloc_c}
    \end{subtable}\\
    \begin{subtable}[h]{0.3\textwidth}
        \centering
        \begin{tabular}{llll}
 & a & b & c \\
\bottomrule
1 & $0$ & $1$ & $0$ \\
2 & $1$ & $0$ & $0$ \\
3 & $0$ & $0$ & $1$ \\
        \bottomrule
       \end{tabular}
       \caption{Profile D}
       \label{tab:alloc_d}
    \end{subtable}
    \begin{subtable}[h]{0.3\textwidth}
        \centering
        \begin{tabular}{llll}
 & a & b & c \\
\bottomrule
1 & $\nicefrac{1}{3}$ & $\nicefrac{1}{2}$ & $\nicefrac{1}{6}$ \\
2 & $\nicefrac{1}{3}$ & $\nicefrac{1}{4}$ & $\nicefrac{5}{12}$ \\
3 & $\nicefrac{1}{3}$ & $\nicefrac{1}{4}$ & $\nicefrac{5}{12}$ \\
        \bottomrule
       \end{tabular}
       \caption{Profile E}
       \label{tab:alloc_e}
    \end{subtable}
    \begin{subtable}[h]{0.3\textwidth}
        \centering
        \begin{tabular}{llll}
 & a & b & c \\
\bottomrule
1 &  & $1$ & $\nicefrac{1}{6}$ \\
2 &  & $0$ &  \\
3 &  & $0$ &  \\
        \bottomrule
       \end{tabular}
       \caption{Profile F}
       \label{tab:alloc_f}
    \end{subtable}
     \caption{Assignments for the six preference profiles in Table \ref{tab:table_profiles}.}
     \label{tab:table_assignment}
\end{table}

\begin{description}

\item [Profile A ($\succ^{(A)}$) :] Consider, first, the preference profile in Table \ref{tab:profile_a}, where 
agent $1$ prefers $a \succ_1^{(A)} b \succ_1^{(A)} c$, agent 2 prefers $b \succ_2^{(A)} a \succ_2^{(A)} c$ while agent 3 prefers $c \succ_3^{(A)} a \succ_3^{(A)} b$.
Since the mechanism $\varphi$ is contention-free efficient, every agent must receive her top choice with probability $1$. That is, $P^{(A)}_{1,a}=P_{2,b}^{(A)}=P_{3,c}^{(A)}=1$.

\item[Profile B ($\succ^{(B)}$) :] Next, suppose agents $1$ and $3$ report their respective preferences as in profile $A$. But agent $2$ swaps her first two choices and reports $a \succ_2^{(B)} b \succ_2^{(B)} c$. First, since strategy-proofness implies lower-invariance, we have $P_{2,c}^{(B)}=P_{2,c}^{(A)}=0$. Second, envy-freeness implies equal treatment of equals. Since agents $1$ and $2$ have the same preferences, both must receive the same random allocation of objects. So, $P_{1,c}^{(B)} = P_{2,c}^{(B)} = 0$ and thus $P_{3,c}^{(B)} = 1$. This implies that  $P_{1,a}^{(B)} = P_{2,a}^{(B)} = \frac{1}{2}$ and $P_{1,b}^{(B)} = P_{2,b}^{(B)} = \frac{1}{2}$.

\item[Profile C ($\succ^{(C)}$) :] If instead agent $2$ reports $a \succ_2^{(C)} c \succ_2^{(C)} b$ while agents $1$ and $3$ report the same preferences as in profile B, upper invariance implies that she should receive an equal probability of being allocated her top object $a$. Therefore, $P_{2,a}^{(C)} = P_{2,a}^{(B)} = \frac{1}{2}$. For agent $1$ to not envy agent $2$, we must have $P_{1,a}^{(C)} = P_{2,a}^{(C)} = \frac{1}{2}$. Consequently,  $ P_{3,a}^{(C)} = 0$. To maintain envy-freeness between agents $2$ and $3$, we have $P_{2,a}^{(C)}+P_{2,c}^{(C)} = P_{3,a}^{(C)}+P_{3,c}^{(C)}$ and $P_{2,b}^{(C)}=P_{3,b}^{(C)}$. Let $P_{2,b}^{(C)}=P_{3,b}^{(C)} =y$. Note than $y \in [0, \frac{1}{2}]$. So we have, $P_{2,c}^{(C)} = \frac{1}{2}-y$ and $P_{3,c}^{(C)}=1-y$.

\item[Profile D ($\succ^{(D)}$) :] Consider the contention-free profile where agents $2$ and $3$ do not change their preferences from profile $C$ but agent $1$ reports $b \succ_1^{(D)} a \succ_1^{(D)} c$. By the assumption that the mechanism is contention-free efficient, each agent receives their top choice with probability $1$. For the mechanism to be strategy-proof, agent $1$ must receive the same probability of receiving object $c$ in profiles $C$ and $D$. That is, $P_{1,c}^{(C)}=P_{1,c}^{(D)} \implies 2y-\frac{1}{2}=0 \implies y=\frac{1}{4}$. 

Thus the resulting assignment for profile $C$ in mechanism $\varphi$ is shown in Table \ref{tab:table_profile_c}.

\begin{table}[h]
    \begin{subtable}[h]{0.45\textwidth}
        \centering
        \begin{tabular}{llll}
\\
\toprule
1 & a & b & c \\
2 & a & c & b \\
3 & c & a & b \\
        \bottomrule
       \end{tabular}
    \end{subtable}
    \begin{subtable}[h]{0.45\textwidth}
        \centering
        \begin{tabular}{llll}
 & a & b & c \\
\bottomrule
1 & $\nicefrac{1}{2}$ & $\nicefrac{1}{2}$ & $0$ \\
2 & $\nicefrac{1}{2}$ & $\nicefrac{1}{4}$ & $\nicefrac{1}{4}$ \\
3 & $0$ & $\nicefrac{1}{4}$ & $\nicefrac{3}{4}$ \\
        \bottomrule
       \end{tabular}
    \end{subtable}
     \caption{Profile C and its final assignment}
     \label{tab:table_profile_c}
\end{table}

\item[Profile E ($\succ^{(E)}$) :]  Suppose, next, that agent $3$ mimics agent $2$ in profile $C$ while agents $1$ and $2$ keep the same preferences as profile $C$. For all three agents to not envy one another, the mechanism must allocate object $a$ equally to all agents. That is, $P_{1,a}^{(E)}=P_{2,a}^{(E)}=P_{3,a}^{(E)} = \frac{1}{3}$. By strategy-proofness, $P_{3,b}^{(E)}=P_{3,b}^{(C)}=\frac{1}{4}$. This implies that $P_{3,c}^{(E)}=\frac{5}{12}$. Since agent $2$ and $3$ have the same preferences, they must get the same allocation. So $P_{2,b}^{(E)}=\frac{1}{4}$ and $P_{2,c}^{(E)}=\frac{5}{12}$. The allocations for the two agents $2$ and $3$ completely determine the allocation that agent $1$ must receive. 

\item[Profile F ($\succ^{(F)}$) :]  Finally, consider the profile in which agent $1$ reports $b \succ_1^{(F)} a \succ_1^{(F)} c$ while agents $2$ and $3$ keep the same preferences as in profile $E$. Since the mechanism is strategy-proof, for agent $3$ to report her preferences truthfully in profiles $F$ and $D$, $P_{3,b}^{(F)}=P_{3,b}^{(D)} = 0$. For agents $2$ and $3$ to not envy each other, $P_{2,b}^{(F)}=0$, which in turn implies $P_{1,b}^{(F)}=1$. However, since the mechanism is strategy-proof, agent $1$ must receive the same probability for object $c$ in profiles $E$ and $F$. That is, $P_{1,c}^{(F)}=P_{1,c}^{(E)}=\frac{1}{6}$, which contradicts feasibility of the allocation.\qedhere
\end{description}
\end{proof}

In order to prove Theorem \ref{thm:main_hardness} for $n>3$, we show that any strategy-proof, envy-free, and contention-free efficient mechanism can used to define an equivalent mechanism for the $n=3$ case. We defer its proof to Appendix \ref{app:hardness}.

When there are three agents and three objects,  \citet{bogomolnaia2001new} mention that strategy-proofness and envy-freeness are incompatible with ex-post efficiency. This was proven formally for $n \geq 3$ by \citet{nesterov2017fairness}. Since ex-post efficiency implies contention-free efficiency, their results arise as an immediate corollary of Theorem \ref{thm:main_hardness}.

\begin{corollary}[{\cite{bogomolnaia2001new,nesterov2017fairness}}] 
For $n\geq 3$, there does not exist a mechanism that is strategy-proof, envy-free and ex-post efficient.
\end{corollary}

When restricted to three agents and objects, we can further strengthen Lemma \ref{theorem:impossibility_=3} to show that strategy-proofness and envy-freeness act as a barrier against all pure deterministic allocations. As discussed earlier, neutrality is a notion of symmetry that requires that a mechanism is invariant to any renaming of objects.  While Lemma \ref{theorem:impossibility_=3} shows that deterministic assignments are impossible in contention-free profiles, the following theorem demonstrates that strategy-proofness and envy-freeness along with neutrality block the pure deterministic allocation of any object to any agent at any profile. We include the proof in Appendix \ref{sec:strong_hardness}.

\begin{theorem} 
\label{thm:strong_hardness}
 Let $|N|=|O|=3$. If $\varphi$ is neutral, strategy-proof and envy-free, then no agent is allocated any object fully at any preference profile under $\varphi$. That is, $\forall \bsucc \in \mathcal{R}^3$, for any agent $i \in N$ and object $a \in O$, $P_{i,a} < 1$, where $P = \varphi(\bsucc)$. 
\end{theorem}

Theorem \ref{thm:main_hardness} demonstrates that in our search for envy-free and strategy-proof mechanisms we must necessarily deviate from (variants of) efficient mechanisms. For example, considering that the RSD mechanism is strategy-proof and weak envy-free, one might try to obtain a strategy-proof and envy-free mechanism by starting from RSD assignments and then successively attempting to remove envy from each assignment. For instance, a similar strategy has been exploited by \citet{harless2019efficient} to show that a variant of the RSD mechanism where randomizing only over the three adjacent agent positions (instead of all orderings) helps recover ordinal efficiency without sacrificing strategy-proofness. But such a strategy is unlikely to yield envy-free and strategy-proof mechanisms considering all natural variants of known efficient mechanisms also maintain contention-free efficiency.
\section{Pairwise Exchange Mechanisms}
\label{sec:separable}

While both strategy-proofness and envy-freeness are highly desirable properties for any random assignment mechanism, the only strategy-proof and envy-free mechanism known prior to this work is the \emph{Equal Division (ED)} mechanism that simply allocates all objects equally to all agents and totally disregards the agents' preferences. In some sense, the equal division mechanism can be considered as the ``most inefficient'' strategy-proof mechanism.
In view of the strong impossibility result presented in Theorem~\ref{thm:main_hardness}, it is natural to wonder whether there exist strategy-proof and envy-free mechanisms that are more efficient than equal division. In this section, we answer this question in the positive and design a large family of mechanisms that are strategy-proof and envy-free and dominate the equal division mechanism. 

\subsection{Transfer Function Representation}
\label{subsec:random_mechanisms_representation}

A natural approach to find more efficient strategy-proof and envy-free mechanisms is to start from the equal division assignment and then allow agents to exchange their allocations depending on their preferences.
In fact, we first argue that, in its most general sense, this view of designing mechanisms by starting from equal allocations and then transferring probabilistic shares of objects between pairs of agents is without loss of generality. In the following proposition, we show that any random assignment mechanism can be represented by a set of \emph{transfer functions} between pairs of agents.

\begin{proposition}\label{prop:matrix_representation}
For every random assignment mechanism $\varphi$, there exist functions $h^{\bsucc}: N \times N \times O \rightarrow [-\frac{1}{n},\frac{1}{n}]$ such that
\begin{align*}
    P^{\bsucc}_{i,a} = \frac{1}{n} + \sum_{j \in N \setminus \{i\}} h^{\bsucc}(i, j, a), \quad \forall \bsucc \in \calR^n, \forall i \in N, \forall a \in O
\end{align*}
where $P^{\bsucc} = \varphi(\bsucc)$.
\end{proposition}

We defer the proof of Proposition \ref{prop:matrix_representation} to Appendix \ref{prop:matrix_representation_proof}. Intuitively, for any profile $\bsucc \in \calR^n$, we interpret $h^{\bsucc}(i,j,a)$ to be the amount of object $a$ that is transferred from agent $j$ to agent $i$. The proposition then asserts that any random assignment mechanism can be viewed as first allocating all objects equally among the agents and then transferring these object shares from one agent to another depending on their reported preferences. 
While in general, these transfers may depend on the identity of the pair of agents as well as the preferences of all agents, it is natural to consider the case where the amount of fractional shares transferred depend solely on the preferences of the pair of agents involved in the exchange.
This motivates us to restrict our attention to mechanisms that admit such a class of transfer functions.

\begin{definition}[Pairwise Exchange Mechanism]
A mechanism $\varphi$ is said to be a \emph{\pairwise} mechanism if there exists a function $f : \cR \times \cR \times O \rightarrow [-\frac{1}{n},\frac{1}{n}]$ such that
\begin{align*}
    P^{\bsucc}_{ia} = \dfrac{1}{n} + \sum_{j \in N \setminus \{i\}} f(\succ_i, \succ_j, a), \quad \forall \bsucc = (\succ_k)_{k \in N} \in \calR^n \text{ and } \forall i,j \in N, a \in O
\end{align*}
where $P^{\bsucc} = \varphi(\bsucc)$. 
We denote the corresponding mechanism by $\varphi^f$. We also drop the superscript when $f$ is clear from context.
\end{definition}

The key restriction we enforce is that $f$ is only a function on the preferences of the two agents involved in the transfer and in particular is independent of the agent identities and the preferences of the other agents. To ensure feasibility of the random assignment, the transfer function $f$ must satisfy certain properties that we tabulate in the following proposition. 
These conditions capture some natural notions and restrictions that one would expect in a bilateral exchange. For instance, the no transfers property states that if a pair of agents have identical preferences over objects, then they would never engage in an exchange of fractional shares for any object. The second property, balanced transfers, ensures the total amount of probabilistic shares that an agent transfers to another agent must equal the amount that she gets back from that same agent. We often refer back to these properties to prove the primary characterization result in the paper. The proposition follows directly from the definition of feasible mechanisms and we include the proof in Appendix \ref{sec:function-properties} for completeness.

\begin{proposition}
\label{prop:necessary_properties}
Consider a \pairwise mechanism $\varphi^f$ and its associated transfer function
$f:\cR \times \cR \times O \rightarrow [-\frac{1}{n}, \frac{1}{n}]$. Then $f$ satisfies the following properties -
\begin{enumerate}
    \item \emph{No transfers:} $f(\succ,\succ,a)=0$ $\forall \succ \in \calR$ and $\forall a \in O$.
    \label{item:identical}
    \item \emph{Balanced transfers:} $\sum_{a \in O} f(\succ, \succ', a) = 0$ for any pair of preferences $\succ, \succ' \in \calR$. \label{item:zero}
    \item \emph{Anti-symmetry:} $f(\succ, \succ', a) = - f(\succ', \succ, a)$ for any pair of preferences $\succ, \succ' \in \calR$ and $\forall a \in O$. \label{item:anti-symmetry}
    \item \emph{Bounded range:} $f(\succ, \succ', a) \in [-\frac{1}{n(n-1)}, \frac{1}{n(n-1)}]$ for any pair of preferences $\succ, \succ' \in \calR$ and $\forall a \in O$. \label{item:sum}
\end{enumerate}
\end{proposition}

As we show in the subsequent sections, \pairwise mechanisms form a rich class of mechanisms and admit many strategy-proof and envy-free mechanisms. In fact, they exhibit an equivalence between fairness and incentive-compatibility. This permits us to focus solely on satisfying either of the two axioms. The restricted form of transfer functions in these mechanisms also makes them amenable to analysis and allows us to exactly characterize the set of all strategy-proof and neutral mechanisms in this class. In addition to their appealing simplicity, these mechanisms capture a natural axiom of \emph{separability} in object allocations.
Consider the allocation received by an agent $i$ in profiles $\bsucc = (\succ_k)_{k \in N}$ and $\bsucc' = (\succ_i, \bsucc'_{-i})$. Informally, we call a mechanism $\varphi$ separable if the difference in these two allocations can be attributed to the difference arising from each agent $j \neq i$ changing her preference from $\succ_j$ to $\succ'_j$.

\begin{definition}
\label{def:separability}
A mechanism $\varphi$ is separable if for all $\bsucc = (\succ_k)_{k \in N} \in \calR^n$, every agent $i \in N$, each $\bsucc'_{-i} = (\succ'_k)_{k \in N \setminus \{i\}} \in \calR^{n-1}$ and every object $a \in O$, we have
\begin{align*}
    P^{(\succ_i, \bsucc'_{-i})}_{i,a} - P^{\bsucc}_{i,a} = \sum_{j \in N \setminus \{i\}} \Big[P^{\bsucc'_j}_{i,a} - P^{\bsucc}_{i,a} \Big]
\end{align*}
where $\bsucc'_j = (\succ'_j, \bsucc_{-j})$, and $P^{\tilde{\bsucc}}=\varphi(\tilde{\bsucc})$ for any $\tilde{\bsucc} \in \calR^n$.
\end{definition}

We now show that the set of anonymous and separable mechanisms is exactly equal to the class of \pairwise mechanisms. We remark that anonymity and separability are both necessary for this characterization since \pairwise mechanisms are anonymous by definition and separability does not imply anonymity\footnote{For example, it is easy to verify that all mechanisms for two agents and objects satisfy Definition \ref{def:separability} but all such mechanisms are not anonymous.}.

\begin{theorem}
A mechanism $\varphi$ is a pairwise exchange mechanism if and only if it is anonymous and separable.
\end{theorem}
\begin{proof}
We prove the easy direction first. Let $\varphi$ be a \pairwise mechanism associated with a transfer function $f$. Anonymity of the mechanism follows from definition since the transfer function $f$ is independent of the identity of the agents. To see separability, we observe that
\begin{align*}
     \sum_{j \in N \setminus \{i\}} \Big[P^{\bsucc'_j}_{i,a} - P^{\bsucc}_{i,a} \Big] &= \sum_{j \in N \setminus \{i\}} \Big[\frac{1}{n} + \sum_{k \in N \setminus \{i,j\}}f(\succ_i,\succ_k,a) + f(\succ_i,\succ'_j,a) - \frac{1}{n} - \sum_{k \in N \setminus \{i\}}f(\succ_i,\succ_k,a)\Big] \\
     &= \sum_{j \in N \setminus \{i\}} f(\succ_i,\succ'_j,a)  - \sum_{j \in N \setminus \{i\}} f(\succ_i,\succ_j,a) \\
     &=P^{\bsucc'_{-i}}_{i,a} - P^{\bsucc}_{i,a}
\end{align*}

To prove the other direction, 
suppose $\varphi$ is anonymous and separable. Let $P^{\bsucc}=\varphi(\bsucc)$ for any $\bsucc \in \calR^n$. For any $\succ,\succ' \in \mathcal{R}$, let $\overline{\bsucc}$ be any preference profile where $(n-1)$ agents have preference $\succ$, and the remaining agent has preference $\succ'$. Let $i$ be one of the $(n-1)$ agents and $j$ be the remaining agent. Define the function $f$ as $f(\succ,\succ',a) := P^{\overline{\bsucc}}_{i,a}-\frac{1}{n}$. Notice by anonymity, function $f$ is well-defined since $P^{\overline{\bsucc}}_{i,a}=P^{\overline{\bsucc}}_{k,a}$ for any agent $k$ other than agent $i$ among the $(n-1)$ agents. Additionally, we also have $0 \leq P^{\overline{\bsucc}}_{i,a} \leq \frac{1}{n-1} \implies -\frac{1}{n} \leq f(\succ,\succ',a) \leq \frac{1}{n}$.

Consider any profile $\bsucc=(\succ_j)_{j \in N}$, any agent $i$ and object $a$. 
Let $\bsucc^{\text{id}} = (\succ_i, \succ_i, \ldots, \succ_i)$ be the preference profile where all agents report the same preference as agent $i$ in $\bsucc$. Let $\bsucc_j = (\succ_j, \bsucc^{\text{id}}_{-j})$ be the resulting preference profile when only agent $j$ changes her report from $\succ_i$ to $\succ_j$. Due to anonymity, $P^{\bsucc^{\text{id}}}_{i,a} = \frac{1}{n}$. By separability we have,
\begin{align*}
    P^{\bsucc}_{i,a}-P^{\bsucc^{\text{id}}}_{i,a} &= \sum_{j \in N \setminus \{i\}} \Big[P^{\bsucc_j}_{i,a} - P^{\bsucc^{\text{id}}}_{i,a} \Big] = \sum_{j \in N \setminus \{i\}} \Big[P^{\bsucc_j}_{i,a} - \frac{1}{n} \Big] \\
    &= \sum_{j  \in N \setminus \{i\}}f(\succ_i,\succ_j,a)\\
    \intertext{which follows from our construction of the function $f$. Therefore we have our desired representation, }
    P^{\bsucc}_{i,a} &= \frac{1}{n} + \sum_{j  \in N \setminus \{i\}}f(\succ_i,\succ_j,a) 
\end{align*}
\end{proof}

At this stage, we would like to contrast our \pairwise mechanisms from the \emph{Top-Trading Cycles from Equal Division (TTCED)} mechanism proposed by \citet{kesten2009popular}. The TTCED mechanism first starts from allocating an equal amount of each object to every agent.  Then, any pair of agents who have received each other's top objects are allowed to exchange their shares sequentially. Once an agent can no longer trade her top object, her allocation of that object is finalized and the trading continues with the next object in the preference. While superficially similar in the sense that both mechanisms start with equal division and allow pairwise trading between agents, the TTCED and \pairwise mechanisms differ in one key aspect - a \pairwise mechanism requires all the trades to occur in parallel independently from each other, while the agent pairs trade sequentially in the TTCED mechanism. Indeed, as shown by \citet{kesten2009popular}, the TTCED mechanism is equivalent to the the PS mechanism of \citet{bogomolnaia2001new} and hence is not strategy-proof. In contrast, we demonstrate that \pairwise mechanisms admit a large set of strategy-proof and envy-free mechanisms.

\subsection{Fairness and Incentives}
\label{subsec:equiv}

In this section, we make the surprising observation that within the class of \pairwise mechanisms, incentive compatibility and fairness go hand in hand. In particular, we demonstrate the equivalence of strategy-proofness and envy-freeness within this class. To the best of our knowledge, this is the only family of random assignment mechanisms for finite markets where such an equivalence between the two normally competing notions of fairness and strategy-proofness has been established\footnote{
Such an equivalence has been observed in large economies (for instance, see \cite{jackson2007envy,che2010asymptotic,liu2016ordinal,noda2018large}).}. Recall that, in \pairwise mechanisms, an agent's total allocation for any profile $\bsucc$ depends on the sum of the transfers of each object received from every other agent. In this setting, strategy-proofness and envy-freeness both necessitate that the transfer function satisfies a natural, but informal, notion of monotonicity - that the transfer between two agents is well-aligned with their corresponding preferences. We formalize this notion through a series of lemmas and defer their proofs to Appendix \ref{app:proof_sp_equiv_ef}. 

The following lemma shows that in a strategy-proof mechanism the transfers of probabilistic shares received by an agent reporting $\succ$ from any other agent stochastically dominate (with respect to $\succ$) the transfers that she receives when reporting any other $\succ'$.
\begin{restatable}{lemma}{spmonotone}
\label{lemma:sp_monotone}
Suppose $\succ = \langle a_1,a_2,\ldots,a_n \rangle$. If $\varphi^f$ is strategy-proof, then for all $\succ', \succ'' \in \calR$ and $t \in [n]$, $\sum_{k=1}^t f(\succ,\succ'',a_k) \geq \sum_{k=1}^t f(\succ',\succ'',a_k)$.
\end{restatable}

Lemma \ref{lemma:sp_swap_upper_lower} shows that strategy-proofness in general random assignment mechanisms is equivalent to the three properties of swap-monotonicity, upper invariance, and lower invariance. Similarly, the following lemma shows that strategy-proofness in \pairwise mechanisms is equivalent to analogous properties of the transfer function.

\begin{restatable}{lemma}{fupperlowerswap}
\label{lemma:f_upper_lower_swap}
A mechanism $\varphi^f$ is strategy-proof if and only if the transfer function $f$ satisfies the following two properties -
\begin{enumerate}
    \item \emph{Sender invariance}: If $\succ' = \langle a_1, a_2, \ldots, a_n\rangle$ and $\succ''$ is such that $\sigma(\succ'', k)=\sigma(\succ',k)=a_k$ $\forall k \leq t$, then $\forall \succ \in \calR$, $f(\succ,\succ',a_t)=f(\succ,\succ'',a_t)$. Similarly, if $\sigma(\succ'', k)=\sigma(\succ',k)=a_k$ $\forall k \geq t$, then $f(\succ,\succ',a_t)=f(\succ,\succ'',a_t)$.\label{item:sender-invariance}
    \item \emph{Swap-monotonicity:} For all $\succ, \succ' \in \calR$ and $\succ'' \in \Gamma(\succ')$ with $a \succ' b$ but $b \succ'' a$ for some $a,b \in O$, we have $f(\succ',\succ, a)\geq f(\succ'',\succ,a)$.
\end{enumerate}
\end{restatable}

Finally, the following two lemmas show that the transfer function for any envy-free mechanism must be ``selfish''. In other words, for any agent, the total amount of objects exchanged with any other agent is compatible with the agent's preferences. Lemma \ref{lemma:ef_zero_transfers} shows that if two agents have identical preferences for their top $t$ objects, then they do not exchange any of those $t$ objects. Similarly, Lemma \ref{lemma:ef_transfer_atleast_zero} shows that the total transfers received by any agent from any other agent stochastically dominates zero transfers.

\begin{restatable}{lemma}{efzerotransfers}
\label{lemma:ef_zero_transfers}
Let $\succ=\langle a_1,a_2,\ldots,a_n\rangle$ be an arbitrary preference relation, and $\succ'$ be such that $\sigma(\succ', k)=\sigma(\succ,k)=a_k$ $\forall k \leq t$. If $\varphi^f$ is envy-free, then $f(\succ,\succ',a_t) = 0$. Similarly, if $\succ'$ is such that $\sigma(\succ', k)=\sigma(\succ,k)=a_k$ $\forall k \geq t$, then $f(\succ,\succ',a_t) = 0$.
\end{restatable}

\begin{restatable}{lemma}{eftransferatleastzero}
\label{lemma:ef_transfer_atleast_zero}
Let $\succ=\langle a_1, a_2, \ldots, a_n\rangle$ be an arbitrary preference relation. If $\varphi^f$ is envy-free, then for all $\succ' \in \calR$, for all $t \in [n]$, $\sum_{k=1}^t f(\succ,\succ',a_t) \geq 0$.
\end{restatable}

The four lemmas above hint towards strategy-proofness and envy-freeness being dependent on slightly different notions of ``monotonicity'' in the transfer functions. Surprisingly, we can in fact show that the two concepts are much more closely tied together. The following theorem demonstrates that a \pairwise mechanism is strategy-proof if and only if it is envy-free. 

\begin{theorem}
\label{thm:sp_equi_envyfree}
In the class of pairwise exchange mechanisms, strategy-proofness is equivalent to envy-freeness.
\end{theorem}
\begin{proof}
We first show that if a pairwise exchange mechanism $\varphi^f$ is strategy-proof, then it is envy-free. 
Suppose, for contradiction, that $\varphi^f$ is strategy-proof but not envy-free. Therefore, there exists some $\bsucc \in \calR$, $i \neq i' \in N$, and some $t \in [n]$ such that the total allocation that agent $i$ gets for her top $t$ objects is strictly smaller than what agent $i'$ receives. Without loss of generality, let $\succ_i = \langle a_1,a_2,\ldots,a_n \rangle$. Since agent $i$ envies $i'$ we have,
\begin{align}
    \sum_{k=1}^tP^{\bsucc}_{i',a_k} &> \sum_{k=1}^tP^{\bsucc}_{i,a_k}
    \intertext{However, since $\varphi^f$ is strategy-proof, at preference profile $\bsucc$, agent $i$ does not gain if she chooses to deviate and report her preference as $\succ'_i=\succ_{i'}$. Let $\bsucc'=(\succ'_i=\succ_{i'},\bsucc_{-i})$ be the resulting preference profile, when agent $i$ mis-reports her preference. Then we have $\sum_{k=1}^tP^{\bsucc}_{i,a_k} \geq \sum_{k=1}^tP^{\bsucc'}_{i,a_k}$. Substituting back into the above inequality, we have}
    \sum_{k=1}^tP^{\bsucc}_{i',a_k} &> \sum_{k=1}^tP^{\bsucc'}_{i,a_k}
    \intertext{Substituting the allocations by their transfer function representation, we get}
    \sum_{k=1}^t f(\succ_{i'},\succ_{i},a_k) + \sum_{j \in N \setminus \{i,i'\}} \sum_{k=1}^t f(\succ_{i'},\succ_j,a_k)
    &> 
    \sum_{k=1}^t f(\succ'_i,\succ_{i'},a_k) + \sum_{j \in N \setminus \{i,i'\}} \sum_{k=1}^t f(\succ'_i,\succ_j,a_k)\\
    &= \sum_{k=1}^t f(\succ_{i'},\succ_{i'},a_k) + \sum_{j \in N \setminus \{i,i'\}} \sum_{k=1}^t f(\succ_{i'},\succ_j,a_k)
    \intertext{By the ``No transfers'' property in Proposition \ref{prop:necessary_properties}, we have that $f(\succ_{i'},\succ_{i'},a_k) = 0, \forall k$, and hence we obtain}
    \sum_{k=1}^t f(\succ_{i'},\succ_{i},a_k) > 0 &= \sum_{k=1}^t f(\succ_{i}, \succ_{i}, a_k)
\end{align}
But this contradicts Lemma \ref{lemma:sp_monotone} when we set $\succ = \succ_i, \succ'=\succ_{i'},  \text{ and } \succ'' = \succ_i$, and hence if $\varphi^f$ is strategy-proof then it must also be envy-free.

We now prove that envy-freeness implies sender invariance and swap monotonicity, and consequently implies strategy-proofness. We prove by contradiction. Suppose that $\varphi^f$ is envy-free, but the transfer function $f$ is not sender invariant. Then, there exists $\succ',\succ'' \in \calR$ such that either $\sigma(\succ', k)=\sigma(\succ'',k)$, $\forall k \leq t$ or $\sigma(\succ', k)=\sigma(\succ'',k)$, $\forall k \geq t$, and a preference $\succ \in \calR$ such that $f(\succ,\succ',\sigma(\succ', t)) \neq f(\succ,\succ'',\sigma(\succ'', t))$. Suppose $\succ'$ is such that $\sigma(\succ', k)=\sigma(\succ'',k)$, $\forall k \leq t$ (the argument for the other case is symmetric). Let $\succ'=\langle a_1, a_2,\ldots, a_n \rangle$. Let $t$ be the smallest such index at which $f(\succ,\succ',a_t) \neq f(\succ,\succ'',a_t)$. Therefore, we have for all $k<t, f(\succ,\succ',a_k) = f(\succ,\succ'',a_k)$. Now, consider a preference profile $\bsucc$, where $\succ_i=\succ'$ for some agent $i$, $\succ_{i'}=\succ''$ for another agent $i' \in N \setminus \{i\}$, and $\succ_j=\succ$ for all other agents $j \in N \setminus \{i,i'\}$. The total allocation that agent $i$ gets for her top $t$ objects is given by,
\begin{align*}
     \sum_{k=1}^t P^{\bsucc}_{i,a_k} &= \frac{t}{n} + \sum_{k=1}^t f(\succ',\succ'',a_k) + (n-2) \sum_{k=1}^t f(\succ',\succ,a_k) \\
     &= \frac{t}{n} + (n-2) \sum_{k=1}^t f(\succ',\succ,a_k) \\
     &\neq \frac{t}{n} + (n-2) \sum_{k=1}^t f(\succ'',\succ,a_k) \\
     &= \frac{t}{n} + \sum_{k=1}^t f(\succ'',\succ',a_k) + (n-2) \sum_{k=1}^t f(\succ'',\succ,a_k) \\
     &= \sum_{k=1}^t P^{\bsucc}_{i',a_k}
\end{align*}
The second and fourth equality above follow from Lemma \ref{lemma:ef_zero_transfers}. Therefore, we find that one of agent $i$ or $i'$ would envy the other, a contradiction.

Next, suppose that $\varphi^f$ is envy-free, but the transfer function $f$ is not swap-monotonic. Then, there exists $\succ' \in \calR$, which we without loss assume to be $\succ'=\langle a_1, a_2, \ldots, a_n \rangle$, $\succ'' \in \Gamma(\succ')$ with $a_t \succ' a_{t+1}$ but $a_{t+1} \succ'' a_t$ for some $t \in [n]$, and $\succ \in \calR$ such that $f(\succ',\succ,a_t)<f(\succ'',\succ,a_t)$. By sender invariance, we have $\sum_{k=1}^{t-1}f(\succ',\succ,a_k)=\sum_{k=1}^{t-1}f(\succ'',\succ,a_k)$ and $\sum_{k=t+2}^{n}f(\succ',\succ,a_k)=\sum_{k=t+2}^{n}f(\succ'',\succ,a_k)$. Since the transfers are balanced (Property (\ref{item:zero}) of Proposition \ref{prop:necessary_properties}), it must be that $f(\succ',\succ,a_{t+1})>f(\succ'',\succ,a_{t+1})$.

Without loss, let $rank(\succ, a_t) < rank(\succ,a_{t+1})$, i.e. $a_t \succ a_{t+1}$ (for the other case, we can just interchange the roles of $\succ'$ and $\succ''$). We can now construct a preference $\tilde{\succ} \in \calR$ such that $rank(\tilde{\succ},a_{t+1})=n$ and for all $s \in [rank(\succ,a_t)], \sigma(\tilde{\succ},s)=\sigma(\succ,s)$. Using sender invariance we have, $f(\succ',\tilde{\succ},a_t)=f(\succ',\succ,a_t) < f(\succ'',\succ,a_t)=f(\succ'',\tilde{\succ},a_t)$. The same property also implies that $\sum_{k=1}^{t-1}f(\succ',\tilde{\succ},a_k)=\sum_{k=1}^{t-1}f(\succ'',\tilde{\succ},a_k)$ and $\sum_{k=t+2}^{n}f(\succ',\tilde{\succ},a_k)=\sum_{k=t+2}^{n}f(\succ'',\tilde{\succ},a_k)$, which in conjunction with the balancedness of $f$ (Property (\ref{item:zero}) of Proposition \ref{prop:necessary_properties}) results in $f(\succ',\tilde{\succ},a_{t+1})>f(\succ'',\tilde{\succ},a_{t+1})$.

Finally, let us construct a preference $\hat{\succ} \in \calR$ such that $rank(\hat{\succ},a_{t+1})=n$ and for all $s \in [t]$, $\sigma(\hat{\succ},s)=\sigma(\succ',s)=a_s$. For such a preference, sender invariance implies that $f(\succ',\hat{\succ},a_{t+1})=f(\succ',\tilde{\succ},a_{t+1})>f(\succ'',\tilde{\succ},a_{t+1})=f(\succ'',\hat{\succ},a_{t+1})$. We again have that $\sum_{k=1}^{t-1}f(\succ',\hat{\succ},a_k)=\sum_{k=1}^{t-1}f(\succ'',\hat{\succ},a_k)$ and $\sum_{k=t+2}^{n}f(\succ',\hat{\succ},a_k)=\sum_{k=t+2}^{n}f(\succ'',\hat{\succ},a_k)$. Therefore, $f(\succ',\hat{\succ},a_{t})<f(\succ'',\hat{\succ},a_{t})$. But, $f(\succ',\hat{\succ},a_{t})=0$ from Lemma \ref{lemma:ef_zero_transfers}, which implies that $f(\succ'',\hat{\succ},a_{t})>0$. Consequently, $f(\hat{\succ},\succ'',a_{t})<0$. But Lemma \ref{lemma:ef_zero_transfers} implies that $f(\hat{\succ},\succ'',a_k)=0$ for all $k \leq (t-1)$. Therefore, we get $\sum_{k=1}^t f(\hat{\succ},\succ'',a_k)<0$, which contradicts Lemma $\ref{lemma:ef_transfer_atleast_zero}$. 
\end{proof}

\subsection{Characterizing Strategy-proof and Neutral Pairwise Exchange Mechanisms}
\label{subsection:sp_neutral_characterize}

We now seek to characterize the set of strategy-proof and neutral \pairwise mechanisms. The following simple lemma shows that the transfer functions for a neutral and strategy-proof \pairwise mechanisms satisfy two additional properties, namely neutrality and receiver invariance, that will be useful in arriving at characterization for this class of mechanisms. The proof follows directly from the definitions of the mechanism and the corresponding properties. Intuitively, if either of these two properties is not satisfied by the transfer function, then we can construct corresponding preference profiles that demonstrate that the mechanism cannot be strategy-proof. We include the proof in Appendix \ref{app:lemma-necessary}.

\begin{lemma}
\label{lem:function-properties-necessary}
If the mechanism $\varphi^f$ is strategy-proof and neutral, then the transfer function $f$ satisfies the properties - 
\begin{enumerate}
    \item \emph{Neutrality:} The transfer function $f : \calR \times \calR \times O \rightarrow [0,1]$ is neutral, i.e. \\
    $f(\succ, \succ', a_i) = f(\pi(\succ), \pi(\succ'), \pi(a_i)) \text{ for any permutation $\pi : O \rightarrow O$}$.
    \label{item:neutral}
    \item \emph{Receiver invariance:} If $\succ = \langle a_1, a_2, \ldots, a_n\rangle$ and $\succ'$ and $\succ''$ are such that $rank(\succ', a_k)=rank(\succ'',a_k)$ $\forall k\leq t$, then $f(\succ,\succ',a_t)=f(\succ,\succ'',a_t)$. Similarly, if $rank(\succ', a_k)=rank(\succ'',a_k)$ $\forall k \geq t$, then $f(\succ,\succ',a_t)=f(\succ,\succ'',a_t)$.\label{item:receiver-invariance}
\end{enumerate}
\end{lemma}

Lemmas \ref{lemma:f_upper_lower_swap} and \ref{lem:function-properties-necessary} show that in our search for strategy-proof and neutral mechanisms, it is enough to restrict ourselves to transfer functions that satisfy neutrality, sender and receiver invariance. We now show these properties, in fact, allow us to prove even stronger restrictions on the transfer function.

\begin{theorem}
\label{thm:main_rank}
If a transfer function $f : \calR \times \calR \times O \rightarrow [-\frac{1}{n},\frac{1}{n}]$ satisfies the properties of neutrality, sender invariance, and receiver invariance, as well as the following:
\[
  \sum_{a \in O} f(\succ, \succ', a) = c \text{  for any pair of preferences $\succ, \succ' \in \calR$, where $c \in \mathbb{R}$ is some constant.}
\]
then for any object $a \in O$, and preferences $\succ, \succ', \succ'' \in \cR$, we have
\[
rank(\succ',a) = rank(\succ'',a) \implies f(\succ, \succ', a) = f(\succ, \succ'', a)
\]
\end{theorem}

\noindent
We prove this claim by induction on $n$. 
For the base case, we consider $n=3$ and include the proof in Appendix \ref{app:theorem-base-case}. Surprisingly, proving the base case also requires a subtle use of the receiver and sender invariance properties, as well as $\sum_{a\in O} f(\succ, \succ', a) = c\ \forall \succ, \succ' \in \calR$.

For the induction step, suppose the claim is true for $n-1 \geq 3$. We need to show that the claim continues to hold for $n$. Let $O=\{a_1,a_2,\dots a_n\}$. Without loss of generality, let $\succ = \langle a_1, a_2, \dots a_n \rangle$. For any $p,q \in [n]$, we define $\calR_{p,q}=\{\succ' \in \calR \mid rank(\succ',a_1)=p \text{ and } rank(\succ',a_n)=q\}$ to be the set of all preferences where the ranks of object $a_1$ and $a_n$ are $p$ and $q$ respectively. Let $R_{p,\ast}=\bigcup_{q \in [n]}R_{p,q}$ and $R_{\ast,q}=\bigcup_{p \in [n]}R_{p,q}$.

\begin{claim}
\label{claim:samepartition_rank_a1}
For all $1 \leq p \leq n$, $\forall \succ', \succ'' \in \calR_{p,\ast}$, and $\forall a \in O,$
\[rank(\succ', a) = rank(\succ'', a) \implies f(\succ, \succ', a) = f(\succ, \succ'', a)\]
\end{claim}

\begin{proof}
By the definition of $\calR_{p,\ast}$, we have $rank(\succ', a_1) = rank(\succ'', a_1) = p$ for any two preferences $\succ', \succ'' \in \calR_{p,\ast}$. Thus, by receiver invariance, we must have $f(\succ, \succ', a_1) = f(\succ, \succ'', a_1)$.

Let $O^{[n-1]} = \{a_2, \ldots, a_{n}\}$ denote the set of $n-1$ objects apart from $a_1$, and let $\calR^{[n-1]}$ denote the set of all strict preferences over $O^{[n-1]}$. 
We can now define an auxiliary function $f': \calR^{[n-1]} \times \calR^{[n-1]} \times O^{[n-1]} \rightarrow [-\frac{1}{n-1},\frac{1}{n-1}]$ as follows
\[f'(\succ, \succ', a) = f(\tilde{\succ}, \tilde{\succ}', a)\ \forall \succ, \succ' \in \calR^{[n-1]}, \forall a \in O^{[n-1]}\]
where $\tilde{\succ}$ is obtained by appending $a_1$ to the beginning of preference relation $\succ$, and $\tilde{\succ}'$ is obtained by inserting $a_1$ at the $p^{\text{th}}$ position in $\succ'$. Now, by receiver invariance and neutrality, $f(\tilde{\succ}, \tilde{\succ}', a_1)$ is equal for any such preferences, and hence we maintain $\sum_{a \in O^{[n-1]}} f'(\succ, \succ', a) = c - f(\tilde{\succ}, \tilde{\succ}', a_1) = c'$ for all $\succ, \succ' \in \cR^{[n-1]}$. 
Further if $\succ = \langle a_2, \ldots, a_n \rangle$, and $\succ'$ and $\succ''$ satisfy $rank(\succ',a_k) = rank(\succ'',a_k)$ for all $ 2 \leq k \leq t$, then we have $f'(\succ, \succ', a_t) = f(\tilde{\succ}, \tilde{\succ}', a_t) = f(\tilde{\succ}, \tilde{\succ}'', a_t) = f'(\succ, \succ'', a_t)$ where the second equality follows from the receiver invariance of function $f$. Thus, the constructed function $f'$ also satisfies receiver invariance.
We can similarly verify that the function $f'$ also satisfies neutrality and sender invariance.
Hence we can use the induction hypothesis  to complete the proof of the claim. 
\end{proof}

We can use a symmetric argument on object $a_n$ to prove the next claim.

\begin{claim}
\label{claim:samepartition_rank_an}
For all $1 \leq q \leq n$, $\forall \succ', \succ'' \in \calR_{\ast,q}$, and $\forall a \in O,$
\[rank(\succ', a) = rank(\succ'', a) \implies f(\succ, \succ', a) = f(\succ, \succ'', a)\]
\end{claim}

In order to complete the proof of the theorem, we now need to argue that the claim is true for any $\succ' \in \calR_{p,q}$ and $\succ'' \in \calR_{s,t}$ where $p\neq s$ and $q \neq t$.

\begin{claim}
\label{claim:differentpartition_rank}
For all $\succ' \in \calR_{p,q}, \succ'' \in \calR_{s,t}$ such that $p \neq s$ and $q \neq t$, and $\forall a \in O,$
\[rank(\succ', a) = rank(\succ'', a) \implies f(\succ, \succ', a) = f(\succ, \succ'', a)\]
\end{claim}

\begin{proof}
Fix $p \neq s$ and $q \neq t$, $\succ' \in \calR_{p,q}, \succ'' \in \calR_{s,t}$ such that $rank(\succ',a)=rank(\succ'',a)$ for some $a \in O$. Note that $a \neq a_n$ and $a \neq a_1$ by definition of the sets $\calR_{p,q}$ and $\calR_{s,t}$. We consider two cases on the object $a$. 

\emph{Case 1: $a \neq a_{n-1}$}

Let $\tilde{\succ}' \in \cR_{\ast,q}$ and $\tilde{\succ}'' \in \cR_{\ast,t}$ be arbitrary preferences such that 
\begin{enumerate}
    \item $rank(\tilde{\succ}',a) = rank(\succ',a)=rank(\tilde{\succ}'',a) = rank(\succ'',a)$
    \item $rank(\tilde{\succ}', a_n) = q$ and $rank(\tilde{\succ}'', a_n) = t$
    \item $rank(\tilde{\succ}',a_i) = rank(\tilde{\succ}'',a_i)$ $\forall i \notin \{n-1,n\}$ 
\end{enumerate}
Note that since $a \neq a_{n-1}$, we are guaranteed to be able to find such preferences, by setting $rank(\tilde{\succ}', a_{n-1}) = rank(\succ'', a_{n}) = t$ and $rank(\tilde{\succ}'', a_{n-1}) = rank(\succ', a_{n}) = q$.
Now, by Claim \ref{claim:samepartition_rank_an}, we must have $f(\succ,\succ',a)=f(\succ,\tilde{\succ}',a)$ and $f(\succ,\succ'',a)=f(\succ,\tilde{\succ}'',a)$. But since $f$ satisfies the receiver invariance property,  we must have $f(\succ,\tilde{\succ}',a)=f(\succ,\tilde{\succ}'',a) \implies f(\succ,\succ',a) = f(\succ,\succ'',a)$ as desired.
\\

\emph{Case 2: $a = a_{n-1}$}

Just as in the case above, our goal is to construct preferences $\tilde{\succ}'$ and $\tilde{\succ}''$ such that $f(\succ, \succ', a) = f(\succ, \tilde{\succ}',a)$ and $f(\succ, \succ'',a) = f(\succ, \tilde{\succ}'',a)$, and then argue that these two transfers must be equal by properties of $\tilde{\succ}'$ and $\tilde{\succ}''$.

Let $\tilde{\succ}' \in \cR_{p,\ast}$ and $\tilde{\succ}'' \in \cR_{s,\ast}$ be arbitrary preferences such that 
\begin{enumerate}
    \item $rank(\succ',a)=rank(\tilde{\succ}',a) = rank(\succ'',a)=rank(\tilde{\succ}'',a)$
    \item $rank(\tilde{\succ}', a_1) = p$ and $rank(\tilde{\succ}'', a_1) = s$
    \item $rank(\tilde{\succ}',a_i) = rank(\tilde{\succ}'',a_i)$ $\forall i \notin \{1,2\}$ 
\end{enumerate}
Again, we are guaranteed to find such preferences. Since $n \geq 4$, we have $a_{n-1} \neq a_2$ and thus we can set $rank(\tilde{\succ}', a_{2}) = rank(\succ'', a_{1}) = s$ and $rank(\tilde{\succ}'', a_{2}) = rank(\succ', a_{1}) = p$.
Now, by Claim \ref{claim:samepartition_rank_a1}, $f(\succ,\succ',a)=f(\succ,\tilde{\succ}',a)$ and $f(\succ,\succ'',a)=f(\succ,\tilde{\succ}'',a)$. Since $f$ satisfies the receiver invariance property,  we must have $f(\succ,\tilde{\succ}',a)=f(\succ,\tilde{\succ}'',a) \implies f(\succ,\succ',a) = f(\succ,\succ'',a)$ as desired. This concludes the proof of the claim.
\end{proof}

Theorem \ref{thm:main_rank} follows as a direct consequence of Claims \ref{claim:samepartition_rank_a1}, \ref{claim:samepartition_rank_an} and \ref{claim:differentpartition_rank}.
We now show that Theorem \ref{thm:main_rank} implies that for any \pairwise mechanism that is neutral and strategy-proof, the transfer function $f$ admits a simple linear representation. The following simple proposition shows that Theorem \ref{thm:main_rank} already implies that the transfer function $f$ admits a simpler form where the transfer of any object $a$ between two agents $i$ and $j$ only depends on the ranks of the object in their corresponding preference orderings. We include the proof in Appendix \ref{app:proposition-proof}.

\begin{proposition}
\label{prop:rankfunction}
For any function $f : \calR \times \calR \times O \rightarrow [-\frac{1}{n},\frac{1}{n}]$, there exists a function $g:[n] \times [n] \rightarrow  [-\frac{1}{n},\frac{1}{n}]$ such that for all $\succ, \succ' \in \calR$ and $\forall a \in O$, we have $f(\succ, \succ', a) = g(rank(\succ, a), rank(\succ', a))$ if and only if $f$ is neutral and for any object $a \in O$, and $\succ, \succ', \succ'' \in \calR$,
\[rank(\succ',a) = rank(\succ'',a) \implies
f(\succ, \succ', a) = f(\succ, \succ'', a) \]
\end{proposition}

Proposition \ref{prop:rankfunction} shows that the transfer function $f$ for any neutral and strategy-proof \pairwise mechanism admits a representation in terms of a function $g:[n] \times [n] \rightarrow [-\frac{1}{n},\frac{1}{n}]$. The following theorem shows that any such function $g$ when combined with the balanced transfers property and the no transfers property must in fact be a linear function. We note that the balanced transfers property implies that for any permutation $\pi$, we have $\sum_{i} g(i, \pi(i)) = 0$, while the no transfers property implies that we have $g(i,i) = 0,\ \forall i \in [n]$.

\begin{theorem}
\label{thm:g-to-linear}
Any function $g:[n] \times [n] \rightarrow \bbR$  satisfies $g(i,i)=0,\ \forall i \in [n]$ and $\sum_{i = 1}^n g(i, \pi(i)) = 0$ for all permutations $\pi:[n] \rightarrow [n]$ if and only if there exists a vector $v \in \bbR^n$ with $v[n]=0$ such that $g(i, j) = v[i] - v[j]$, where $v[k]$ is the $k^{\text{th}}$  coordinate of $v$.
\end{theorem}

\begin{proof}
Let us prove the easy direction first. Suppose $\exists \ v \in \bbR^n$ s.t. $g(i,j)=v[i]-v[j]$. First, for any $i \in [n]$, $g(i,i)=v[i]-v[i] = 0$. Second, for any permutation $\pi$,
\begin{align*}
    \sum_{i = 1}^n g(i, \pi(i)) &= \sum_{i = 1}^n \Big(v[i]-v[\pi(i)]\Big) 
    = \sum_{i = 1}^n v[i] - \sum_{i = 1}^n v[\pi(i)] 
    = 0
\end{align*}
since $\pi$ is a permutation and hence a bijection on $[n]$.

For the other direction of the claim, we define the vector $v \in \bbR^n$ as, 
\begin{align*}
    v[i] = g(i,n) \ \forall i \in [n]
\end{align*}

Let $\pi$ be any permutation and let $\pi'$ be such that it is identical to $\pi$ except for a swap on one pair of indices. That is, $\pi, \pi'$ are such that $\pi(i)=\pi'(i)$ for all $i \in [n] \setminus \{m,k\}$, $\pi(m)=\pi'(k)$ and $\pi(k)=\pi'(m)$ for some $m, k \in [n]$. Since the function $g$ satisfies $\sum_{i = 1}^n g(i, \pi(i)) = 0$ for all permutations $\pi$, we have: 
\begin{align*}
    \sum_{i = 1}^n g(i, \pi(i)) &= \sum_{i = 1}^n g(i, \pi'(i)) \\
    \text{i.e.} \sum_{i \in [n]\setminus\{m,k\}} g(i, \pi(i)) + g(m, \pi(m)) + g(k, \pi(k)) &= \sum_{i \in [n]\setminus\{m,k\}} g(i, \pi'(i)) + g(m, \pi'(m)) + g(k, \pi'(k)) \\
    \text{i.e. } g(m, \pi(m)) + g(k, \pi(k)) &= g(m, \pi(k)) + g(k, \pi(m)) 
\end{align*}

Thus for any $m,k \in [n]$, 
if we consider a particular permutation $\pi$ such that $\pi(m)=k$ and $\pi(k)=n$, the above equality reduces to 
\begin{align}
g(m, k) + g(k, n) &= g(m, n) + g(k,k)
\intertext{But since $g(k,k) = 0$, this yields}
 g(m, k) &= g(m, n) - g(k, n) = v[m] - v[k] \label{eq:linearswap}
\end{align}
as desired.
\end{proof}

Finally, combining Theorem \ref{thm:sp_equi_envyfree}, Lemma \ref{lem:function-properties-necessary}, Theorem \ref{thm:main_rank}, Proposition \ref{prop:rankfunction}, and Theorem \ref{thm:g-to-linear}, we obtain our principle characterization result.

\begin{theorem}
\label{thm:main-charac}
A \pairwise mechanism $\varphi^f$ is strategy-proof, envy-free, and neutral, if and only if there exists a vector $v \in [0,\frac{1}{n(n-1)}]^n$ with $v[k]>=v[k+1]$ $\forall k \in [n-1], v[n] = 0$, such that $\forall \succ, \succ' \in \calR$, $\forall a \in O$, $f(\succ,\succ',a) = v[rank(\succ,a)] - v[rank(\succ',a)]$, where $v[k]$ is the $k^{\text{th}}$  coordinate of $v$.
\end{theorem}
\begin{proof}[Proof Sketch]
Lemma \ref{lem:function-properties-necessary} to Theorem \ref{thm:g-to-linear} together imply that the transfer function for any neutral, strategy-proof, and envy-free \pairwise exchange mechanism must admit a simple linear representation. The additional restrictions on the vector follow from conditions necessary for the mechanism to be feasible. We include the formal proof of the theorem in Appendix \ref{app:thm-pairwise-charac}.
\end{proof}

We denote a strategy-proof, envy-free, and neutral \pairwise mechanism by $\varphi^v$, where $v$ is its associated vector in its representation in Theorem \ref{thm:main-charac}.

\subsection{Efficiency in Pairwise Exchange Mechanisms}
\label{sec:pareto-dominance}

Efficiency is an important consideration in designing house allocation mechanisms. Lemma \ref{lemma:ef_transfer_atleast_zero} shows that any envy-free pairwise exchange mechanism is at least as efficient as the equal division mechanism and thus our characterization demonstrates a large class of strategy-proof and envy-free mechanisms that all (weakly) dominate equal division.

Formally, a mechanism $\varphi$ weakly dominates another mechanism $\varphi'$, if for every profile $\bsucc \in \calR^n$ and any agent $i \in N$, $P_i \geq_i P'_i$, where $P=\varphi(\bsucc)$ and $P'=\varphi'(\bsucc)$. If in addition there exists a strict domination for atleast one preference profile, then we say $\varphi$ (strictly) dominates $\varphi'$. Let $\mathcal{F}$ denote a family of mechanisms. We say that a mechanism $\varphi \in \mathcal{F}$ is Pareto-efficient in $\mathcal{F}$ if no other mechanism $\varphi' \in \mathcal{F}$ dominates $\varphi$. 

It is easy to see from Theorem \ref{thm:main-charac} that if $v[1] > 0$ in the associated vector $v$, then the mechanism $\varphi^v$ (strictly) dominates the equal division mechanism. A question that we consider next is to find the set of mechanisms that are most efficient within the family of linear mechanisms. The following theorem gives a characterization of the set of Pareto-efficient linear mechanisms. We defer the proof of the following theorem to Appendix \ref{app:linear-pareto-efficient}.

\begin{theorem}
\label{thm:linear-pareto-efficient}
Let $\cF$ be the class of strategy-proof, envy-free, and neutral \pairwise mechanisms. A mechanism $\varphi^v \in \cF$ is Pareto-efficient in $\cF$ if and only the vector $v$ satisfies
$v[1] = \frac{1}{n(n-1)}$.
\end{theorem}

\section{Conclusion}
\label{sec:conclusions}
In this work, we focused on strategy-proof and envy-free random assignment mechanisms for the house allocation problem. We first defined a very weak notion of efficiency, called contention-free efficiency, and showed that within the class of strategy-proof mechanisms, envy-freeness is incompatible even with this weak form of efficiency. Next, we considered a broad class of mechanisms called Pairwise Exchange mechanisms, provided an axiomatic characterization for this class and showed that strategy-proofness is equivalent to envy-freeness within this class. In a \pairwise mechanism, agents are initially allocated all objects equally and then every pair of agents exchange their corresponding shares, where the amount of any object exchanged is dictated by a transfer function that depends only on the preferences of agents involved in the exchange. We characterized the set of all strategy-proof, envy-free, and neutral mechanisms within this class and showed that such mechanisms admit a simple linear representation.

Given the trade-offs between incentives, fairness and efficiency, characterizing the Pareto efficient frontier within the class of strategy-proof and envy-free mechanisms is an important question. Indeed, by restricting transfer functions to depend only on the preferences of the agents involved in the exchange, the mechanisms we characterize do not allow us to reach this frontier. However, we consider our work as a first step towards answering that question. In future work, we plan to study more general transfer functions and identify the structure of strategy-proof and envy-free mechanisms admitting such functions.

\bibliography{bibliography}
\begin{appendix}
\section{Proof of Theorem \ref{thm:main_hardness} for $n>3$}
\label{app:hardness}

Let $N=\{1,2,3, \dots n\}$ and $O=\{a_1, a_2, a_3, \dots a_n\}$ denote the set of agents and objects respectively. Let us again suppose for contradiction that there exists a mechanism $\varphi$ that is strategy-proof, envy-free and contention-free efficient.  We consider a family of preference profiles $\mathcal{F}^n$, where agents $1$, $2$ and $3$ prefer objects $a_1$, $a_2$ and $a_3$, in any order, to $\langle a_4, a_5, \dots a_n\rangle$, while every agent $i>3$ has the preference relation  $\succ_i= \langle a_i,\ldots a_n, a_1, a_2, \ldots, a_{i-1} \rangle$. Table \ref{tab:familyprofiles} illustrates this family of profiles. The $\star$ in the table refers to any ordering of objects $a_1, a_2,$ and $a_3$. 

\begin{table}[h]
        \centering
        \begin{tabular}{llllllllll}
        \\
        \toprule
1 &  &  &  & $a_4$ & $a_5$ & \ldots & \ldots & \ldots & $a_n$ \\
2 &  & $\star$ &  & $a_4$ & $a_5$ & \ldots & \ldots & \ldots & $a_n$ \\
3 &  &  &  & $a_4$ & $a_5$ & \ldots & \ldots & \ldots & $a_n$ \\
4 & $a_4$ & $a_5$ & $a_6$ & \ldots  & \ldots & \ldots & $a_1$ &$a_2$ &$a_3$ \\
5 & $a_5$ & $a_6$ & $a_7$ & \ldots & \ldots & \ldots & $a_2$ & $a_3$ & $a_4$ \\
\vdots & \vdots & \vdots & \vdots & \vdots  & \vdots &\vdots & \vdots & \vdots & \vdots \\
n & $a_n$ & $a_1$ & $a_2$ & \ldots & \ldots & \ldots & $a_{n-3}$ & $a_{n-2}$ & $a_{n-1}$ \\
        \bottomrule
       \end{tabular}
       \caption{Family of Preference Profiles $\mathcal{F}^n$}
       \label{tab:familyprofiles}
\end{table}

We first prove the following lemma to show that if the mechanism $\varphi$ satisfies all the three premises of the theorem, then in every preference profile in $\mathcal{F}^n$ agents $1, 2,$ and $3$ are assigned a zero probability of receiving any object from $\{a_4, a_5, \ldots, a_n\}$.

\begin{lemma}\label{reduction_lemma}
If mechanism $\varphi$ is strategy-proof, envy-free and contention-free efficient, then for every $\boldsymbol{\succ} \in \mathcal{F}^n$, $P_{i,a_j}=0$ for $i \in \{1,2,3\}$ and $j \in \{4,\dots n\}$ where $P = \varphi(\bsucc)$.
\end{lemma}

\begin{proof}
Consider any $\boldsymbol{\succ}=(\succ_1,\succ_2,\succ_3,\boldsymbol{\succ}_{N\setminus \{1,2,3\}}) \in \mathcal{F}^N$. Consider also a preference profile $\boldsymbol{\succ}^{(0)}= (\succ_1^{(0)},\succ_2^{(0)},\succ_3^{(0)},\boldsymbol{\succ}_{N\setminus \{1,2,3\}})\in \mathcal{F}^N$ where agents $1, 2, $ and $3$ prefer $a_1, a_2, $ and $a_3$ as their top choice respectively. Notice that $\boldsymbol{\succ}^{(0)}$ is a contention-free preference profile. Since $\varphi$ is contention-free efficient, $P_{1,a_1}^{(0)}=P_{2,a_2}^{(0)}=P_{3,a_3}^{(0)}=1$, which implies that $P_{i,a_j}^{(0)}=0$ for $i \in \{1,2,3\}$ and $j \in \{4,5,\dots n\}$.

We now consider a sequence of preference profiles \{$\boldsymbol{\succ}^{(1)}$,$\boldsymbol{\succ}^{(2)}$,$\boldsymbol{\succ}^{(3)}=\boldsymbol{\succ}$\} where agents $1, 2, $ and 3 successively report their preference in $\boldsymbol{\succ}$. That is, $\boldsymbol{\succ}^{(1)}= (\succ_1,\succ_2^{(0)},\succ_3^{(0)},\boldsymbol{\succ}_{N\setminus \{1,2,3\}})$,  $\boldsymbol{\succ}^{(2)}= (\succ_1,\succ_2,\succ_3^{(0)},\boldsymbol{\succ}_{N\setminus \{1,2,3\}})$, $\boldsymbol{\succ}^{(3)}= \bsucc = (\succ_1,\succ_2,\succ_3,\boldsymbol{\succ}_{N\setminus \{1,2,3\}})$. Observe that between $\boldsymbol{\succ}^{(0)}$ and $\boldsymbol{\succ}^{(1)}$, only agent $1$'s preferences differ. But in both cases, agent $1$ prefers objects $\{a_1, a_2, a_3\}$ over all the other objects. Since $\varphi$ is strategy-proof, it must be that 
\begin{align*}
    P_{1,a_1}^{(1)}+P_{1,a_2}^{(1)}+P_{1,a_3}^{(1)}&=P_{1,a_1}^{(0)}+P_{1,a_2}^{(0)}+P_{1,a_3}^{(0)}=1 \\
    \implies P_{1,a_j}^{(1)}&=0 \text{ for } j \in \{4,5,\dots n\}
\end{align*}
Now, since $\varphi$ is also envy-free, for agents $1, 2, $ and $3$ to not envy each other, we must have,
\begin{align*}
    P_{2,a_1}^{(1)}+P_{2,a_2}^{(1)}+P_{2,a_3}^{(1)}&=P_{3,a_1}^{(1)}+P_{3,a_2}^{(1)}+P_{3,a_3}^{(1)} 
    =P_{1,a_1}^{(1)}+P_{1,a_2}^{(1)}+P_{1,a_3}^{(1)}
    =1 \\
    \implies P_{2,a_j}^{(1)}=P_{3,a_j}^{(1)} &= 0 \text{ for } j \in \{4,5,\dots n\}
\end{align*}
We can apply the same arguments as we move from $\boldsymbol{\succ}^{(1)}$ to $\boldsymbol{\succ}^{(2)}$ where only agent $2$ changes her report and from  $\boldsymbol{\succ}^{(2)}$ to $\boldsymbol{\succ}^{(3)}$ where only agent $3$ reports different preferences. In each instance, it will always be that $P_{i,a_j}=0$ for $i \in \{1,2,3\}$ and $j \in \{4,5,\dots n\}$.
\end{proof}

Lemma \ref{reduction_lemma} implies that the assignment problem for $n>3$ can be reduced to a problem on the first three agents. Since $\varphi$ satisfies the premises of the theorem, we can define a new mechanism $\varphi'$ for $n=3$ by restricting to the domain of $\mathcal{F}^n$. Formally, let $\boldsymbol{\succ'}=(\succ'_1,\succ'_2,\succ'_3) \in \mathcal{R}^3$ be a preference profile for the first three agents in the reduced problem and let $\boldsymbol{\succ}=(\succ_1,\succ_2,\succ_3, \boldsymbol{\succ}_{N\setminus\{1,2,3\}}) \in \mathcal{F}^n$ be a preference profile for $n$ agents such that each agent $i \in \{1,2,3\}$ has the preference relation with objects $a_1, a_2,$ and $a_3$ preferred according to $\succ'_i$ followed by $\langle a_4, a_5,\dots a_n\rangle$. In the random assignment $P=\varphi(\boldsymbol{\succ})$, Lemma \ref{reduction_lemma} implies that $P_{i,a_j}=0$ for $i \in \{1,2,3\}$ and $j \in \{4,5,\dots n\}$. Therefore, we can define $P'=\varphi(\boldsymbol{\succ'})$ as $P'=[P_i]_{i \in \{1,2,3\}}$, which is a valid random assignment. However, since $\varphi$ is contention-free, envy-free, and strategy-proof, $\varphi'$ also satisfies all the three properties and thus contradicts Lemma \ref{theorem:impossibility_=3}. This completes the proof of Theorem~\ref{thm:main_hardness}.

\section{Proof of Theorem \ref{thm:strong_hardness}}
\label{sec:strong_hardness}

Let $N=\{1,2,3\}$ and $O=\{a,b,c\}$. Suppose for contradiction that the mechanism fully allocates an object to an agent at some profile. Let us suppose, without loss, that this is agent $1$ and her preference at this profile is $a \succ_1 b \succ_1 c$. It must be that the object fully allocated to agent $1$ is her top choice, i.e., object $a$, otherwise she will envy either of the other two agents who will receive a non-zero probability of being assigned object $a$. Also, in this profile agents $2$ and $3$ cannot have object $a$ as their most preferred, as any resulting assignment would violate envy-freeness. Therefore, we need to consider two cases based on the objects most preferred by agents $2$ and $3$.

\emph{Case 1: Agents $2$ and $3$ prefer the same object, $b$ or $c$, as their top choice.}

Since both agents prefer the same object, it must be that both have object $a$ as their least preferred as otherwise any resulting assignment would not envy-free for either agent $2$ or $3$. We, first, consider the sub-case where both agents have preference $b \succ c \succ a$. We call this preference profile, Profile $A$. 

\emph{Case 1a: Agents $2$ and $3$  prefer $b \succ c \succ a$.}

Beginning with Profile $A$, we consider a series of eight profiles that will lead to us a contradiction. The profiles and their corresponding assignments are shown in Table $\ref{tab:apptable_profiles}$ and $\ref{tab:apptable_assignment}$ respectively.

\begin{table}[h]
    \begin{subtable}[h]{0.23\textwidth}
        \centering
        \begin{tabular}{llll}
        \\
        \toprule
1 & a & b & c \\
2 & b & c & a \\
3 & b & c & a \\
        \bottomrule
       \end{tabular}
       \caption{Profile A}
       \label{tab:appprofile_a}
    \end{subtable}
    \begin{subtable}[h]{0.23\textwidth}
        \centering
        \begin{tabular}{llll}
        \\
        \toprule
1 & a & c & b \\
2 & b & c & a \\
3 & b & c & a \\
        \bottomrule
       \end{tabular}
       \caption{Profile B}
       \label{tab:appprofile_b}
    \end{subtable}
    \begin{subtable}[h]{0.23\textwidth}
        \centering
        \begin{tabular}{llll}
        \\
        \toprule
1 & c & a & b \\
2 & b & c & a \\
3 & b & c & a \\
        \bottomrule
       \end{tabular}
       \caption{Profile C}
       \label{tab:appprofile_c}
    \end{subtable}
    \begin{subtable}[h]{0.23\textwidth}
        \centering
        \begin{tabular}{llll}
        \\
        \toprule
1 & c & b & a \\
2 & b & c & a \\
3 & b & c & a \\
        \bottomrule
       \end{tabular}
       \caption{Profile D}
       \label{tab:appprofile_d}
    \end{subtable}\\
    \begin{subtable}[h]{0.23\textwidth}
        \centering
        \begin{tabular}{llll}
        \\
        \toprule
1 & a & b & c \\
2 & b & a & c \\
3 & b & c & a \\
        \bottomrule
       \end{tabular}
       \caption{Profile E}
       \label{tab:appprofile_e}
    \end{subtable}
    \begin{subtable}[h]{0.23\textwidth}
        \centering
        \begin{tabular}{llll}
        \\
        \toprule
1 & c & b & a \\
2 & b & c & a \\
3 & b & a & c \\
        \bottomrule
       \end{tabular}
       \caption{Profile F}
       \label{tab:appprofile_f}
    \end{subtable}
    \begin{subtable}[h]{0.23\textwidth}
        \centering
        \begin{tabular}{llll}
        \\
        \toprule
1 & c & a & b \\
2 & c & b & a \\
3 & b & c & a \\
        \bottomrule
       \end{tabular}
       \caption{Profile G}
       \label{tab:appprofile_g}
    \end{subtable}    \begin{subtable}[h]{0.23\textwidth}
        \centering
        \begin{tabular}{llll}
        \\
        \toprule
1 & a & c & b \\
2 & c & b & a \\
3 & b & c & a \\
        \bottomrule
       \end{tabular}
       \caption{Profile H}
       \label{tab:appprofile_h}
    \end{subtable}
    \caption{Eight Preference Profiles In Case (1a)}
     \label{tab:apptable_profiles}
\end{table}

\begin{table}[h]
    \begin{subtable}[h]{0.3\textwidth}
        \centering
        \begin{tabular}{llll}
 & a & b & c \\
\bottomrule
1 & 1 & 0 & 0 \\
2 & 0 & \nicefrac{1}{2} & \nicefrac{1}{2} \\
3 & 0 & \nicefrac{1}{2} & \nicefrac{1}{2} \\
        \bottomrule
       \end{tabular}
       \caption{Profile A}
       \label{tab:appalloc_a}
    \end{subtable}
    \begin{subtable}[h]{0.3\textwidth}
        \centering
        \begin{tabular}{llll}
 & a & b & c \\
\bottomrule
1 & 1 & 0 & 0 \\
2 & 0 & \nicefrac{1}{2} & \nicefrac{1}{2} \\
3 & 0 & \nicefrac{1}{2} & \nicefrac{1}{2} \\
        \bottomrule
       \end{tabular}
       \caption{Profile B}
       \label{tab:appalloc_b}
    \end{subtable}
    \begin{subtable}[h]{0.3\textwidth}
        \centering
        \begin{tabular}{llll}
 & a & b & c \\
\bottomrule
1 & $1-x$ & 0 & $x$ \\
2 & $\nicefrac{x}{2}$ & $\nicefrac{1}{2}$ & $\nicefrac{1}{2}-\nicefrac{x}{2}$ \\
3 & $\nicefrac{x}{2}$ & $\nicefrac{1}{2}$ & $\nicefrac{1}{2}-\nicefrac{x}{2}$ \\
        \bottomrule
       \end{tabular}
       \caption{Profile C}
       \label{tab:appalloc_c}
    \end{subtable}\\
    \begin{subtable}[h]{0.3\textwidth}
        \centering
        \begin{tabular}{llll}
 & a & b & c \\
\bottomrule
1 & $\nicefrac{1}{3}$ & $\nicefrac{2}{3}-x$ & $x$ \\
2 & $\nicefrac{1}{3}$ & $\nicefrac{1}{6}+\nicefrac{x}{2}$ & \\
3 & $\nicefrac{1}{3}$ & $\nicefrac{1}{6}+\nicefrac{x}{2}$ & \\
        \bottomrule
       \end{tabular}
       \caption{Profile D}
       \label{tab:appalloc_d}
    \end{subtable}
    \begin{subtable}[h]{0.3\textwidth}
        \centering
        \begin{tabular}{llll}
 & a & b & c \\
\bottomrule
1 & $1-y$ & $0$ & $y$ \\
2 & & $\nicefrac{1}{2}$ & $y$ \\
3 & & $\nicefrac{1}{2}$ & \\
        \bottomrule
       \end{tabular}
       \caption{Profile E}
       \label{tab:appalloc_e}
    \end{subtable}
    \begin{subtable}[h]{0.3\textwidth}
        \centering
        \begin{tabular}{llll}
 & a & b & c \\
\bottomrule
1 & $y$ & $0$ & $1-y$ \\
2 & $y$ & $\nicefrac{1}{2}$ & \\
3 & & $\nicefrac{1}{2}$ & \\
        \bottomrule
       \end{tabular}
       \caption{Profile F}
       \label{tab:appalloc_f}
    \end{subtable}
    \begin{subtable}[h]{0.3\textwidth}
        \centering
        \begin{tabular}{llll}
 & a & b & c \\
\bottomrule
1 & $\nicefrac{1}{3}$ & $\nicefrac{2}{3}-z$ & $z$ \\
2 & $\nicefrac{1}{3}$ & $\nicefrac{2}{3}-z$ & $z$ \\
3 & $\nicefrac{1}{3}$ &  & \\
        \bottomrule
       \end{tabular}
       \caption{Profile G}
       \label{tab:appalloc_g}
    \end{subtable}
    \begin{subtable}[h]{0.3\textwidth}
        \centering
        \begin{tabular}{llll}
 & a & b & c \\
\bottomrule
1 & $1$ & $\nicefrac{2}{3}-z$ & \\
2 & $0$ & &  \\
3 & $0$ & &  \\
        \bottomrule
       \end{tabular}
       \caption{Profile H}
       \label{tab:appalloc_h}
    \end{subtable}
     \caption{Assignment for the Eight Preference Profiles in Case (1b)}
     \label{tab:apptable_assignment}
\end{table}

\begin{description}
\item[Profile A ($\succ^{(A)}$):]
In this profile, we start off with $P^{(A)}_{1,a}=1$. Since envy-freeness implies equal treatment of equals, the only possible allocation for agents $2$ and $3$ is $P^{(A)}_{2,b}=P^{(A)}_{2,c}=P^{(A)}_{3,b}=P^{(A)}_{3,c}=\frac{1}{2}$.

\item[Profile B ($\succ^{(B)}$):]
Next, suppose agent $1$ reports $a \succ^{(B)}_1 c \succ^{(B)}_1 b$ whereas agents $2$ and $3$ maintain same preferences as in Profile A. For agent $1$ to not gain by mis-reporting her preferences between Profile B and A, we must have $P^{(B)}_{1,a}=1$. Again, by equal treatment of equals, $P^{(A)}_{2,b}=P^{(A)}_{2,c}=P^{(A)}_{3,b}=P^{(A)}_{3,c}=\frac{1}{2}$.

\item[Profile C ($\succ^{(C)}$):]
If agent $1$ chooses to instead report $c \succ^{(C)}_1 a \succ^{(C)}_1 b$, then strategy-proofness between Profiles B and C implies that $P^{(C)}_{1,b}=P^{(B)}_{1,b}=0$. Let $P^{(C)}_{1,c}=x$, where $x \in [0,1]$.

\item[Profile D ($\succ^{(D)}$):]
 Agent $1$ now reports $c \succ^{(D)}_1 b \succ^{(D)}_1 a$ while agents $2$ and $3$ report same preferences as before. For all three agents to not envy one another, they must receive the same probability for being assigned object $a$. Therefore, we must have $P^{(D)}_{1,a}=P^{(D)}_{2,a}=P^{(D)}_{3,a}=\frac{1}{3}$. Next, strategy-proofness between Profiles C and D implies that $P^{(D)}_{1,c}=P^{(C)}_{1,c}=x$. So we have $P^{(D)}_{1,b}=\frac{2}{3}-x$. Since agents $2$ and $3$ have identical preferences, they must receive identical allocations. Therefore, $P^{(D)}_{2b}=P^{(D)}_{3b}=\frac{1}{6}+\frac{x}{2}$.
 
\item[Profile E ($\succ^{(E)}$):]
Suppose now, agent $2$ changes her report to $b \succ^{(E)}_2 a \succ^{(E)}_2 c$ in Profile A. For agent $2$ to report truthfully, $P^{(E)}_{2,b}=P^{(A)}_{2,b}=\frac{1}{2}$. Envy-freeness implies $P^{(E)}_{3,b}=P^{(E)}_{2,b}=\frac{1}{2}$. Consequently, $P^{(E)}_{1,b}=0$. Let $P^{(E)}_{1,c}=y$, where $y \in [0,1]$. For agents $1$ and $2$ to not envy each other, they must get the same probability of receiving their least preferred object $c$. Therefore, $P^{(E)}_{2,c}=y$.

\item[Profile F ($\succ^{(F)}$):]
Now let's consider a profile where agents $1$, $2$, and $3$ prefer $c \succ^{(F)}_1 b \succ^{(F)}_1 a$, $b \succ^{(F)}_2 c \succ^{(F)}_2 a$, and $b \succ^{(F)}_3 a \succ^{(F)}_3 c$ respectively. Observe that Profile F can be obtained from Profile E by re-labeling objects $a$ and $c$. Since the mechanism satisfies neutrality, we must have $P^{(F)}_{i,b}=P^{(E)}_{i,b}$ $\forall i \in \{1,2,3\}$. But if we compare Profiles D and F, only agent $3$ has changed her report. By upper-invariance, $P^{(D)}_{3,b}=P^{(F)}_{3,b}=\frac{1}{2} \implies \frac{1}{6} + \frac{x}{2}= \frac{1}{2} \implies x =\frac{2}{3}$.

Thus, the resulting assignment for Profile C is presented in Table \ref{tab:apptable_profile_c}.

\begin{table}[h]
    \begin{subtable}[h]{0.45\textwidth}
        \centering
        \begin{tabular}{llll}
\\
\toprule
1 & c & a & b \\
2 & b & c & a \\
3 & b & c & a \\
        \bottomrule
       \end{tabular}
    \end{subtable}
    \begin{subtable}[h]{0.45\textwidth}
        \centering
        \begin{tabular}{llll}
 & a & b & c \\
\bottomrule
1 & $\nicefrac{1}{3}$ & 0 & $\nicefrac{2}{3}$ \\
2 & $\nicefrac{1}{3}$ & $\nicefrac{1}{2}$ & $\nicefrac{1}{6}$ \\
3 & $\nicefrac{1}{3}$ & $\nicefrac{1}{2}$ & $\nicefrac{1}{6}$ \\
        \bottomrule
       \end{tabular}
    \end{subtable}
     \caption{Profile C}
     \label{tab:apptable_profile_c}
\end{table}

\item[Profile G ($\succ^{(G)}$):]
Suppose agent $2$ mis-reports her preferences in Profile C with $c \succ^{(G)}_2 b \succ^{(G)}_2 a$. Lower-invariance implies that $P^{(G)}_{2,a}=P^{(C)}_{2,a}=\nicefrac{1}{3}$. For agent $3$ to not envy agent $2$, $P^{(G)}_{3,a}=\nicefrac{1}{3}$. So $P^{(G)}_{1,a}=\nicefrac{1}{3}$. Let $P^{(G)}_{1,c}=z$. Note that $z \in [\frac{1}{3},\frac{1}{2}]$.

\item[Profile H ($\succ^{(H)}$):]
Finally, suppose agent $1$ changes her report to $a \succ^{(H)}_1 c \succ^{(H)}_1 b$ in Profile G. By lower-invariance, $P^{(H)}_{1,b}=P^{(G)}_{1,b}=\frac{2}{3}-z$. But if we compare this Profile with Profile B, only agent $2$'s preferences change. Strategy-proofness implies that $P^{(H)}_{2,a}=P^{(B)}_{2,a}=0$. Further, envy-freeness between agents $2$ and $3$ implies that $P^{(H)}_{3,a}=P^{(H)}_{2,a}=0 \implies P^{(H)}_{1,a}=1$. So $P^{(H)}_{2,b}=0 \implies \frac{2}{3}-z=0 \implies z=\frac{2}{3}$. This is a contradiction to the condition necessary to maintain feasibility in Profile G.
\end{description}

\emph{Case 1b: Agents $2$ and $3$  prefer $c \succ b \succ a$.}

In this case, the unique envy-free assignment, with agent $1$ receiving object $a$ with probability one, is shown in Table \ref{tab:apptable_case1b}. Since the mechanism is neutral, re-labeling objects $b$ and $c$ leads us back to Profile B and its corresponding assignment. Strategy-proofness and envy-freeness between Profiles A and B in Case 1a, in turn, imply that the mechanism must allocate the same assignment to Profile A. We can now use Case 1a to arrive at a contradiction. 

\begin{table}[h]
    \begin{subtable}[h]{0.45\textwidth}
        \centering
        \begin{tabular}{llll}
\\
\toprule
1 & a & b & c \\
2 & c & b & a \\
3 & c & b & a \\
        \bottomrule
       \end{tabular}
    \end{subtable}
    \begin{subtable}[h]{0.45\textwidth}
        \centering
        \begin{tabular}{llll}
 & a & b & c \\
\bottomrule
1 & $1$ & $0$ & $0$ \\
2 & $0$ & $\nicefrac{1}{2}$ & $\nicefrac{1}{2}$ \\
3 & $0$ & $\nicefrac{1}{2}$ & $\nicefrac{1}{2}$ \\
        \bottomrule
       \end{tabular}
    \end{subtable}
     \caption{Profile for Case 1b}
     \label{tab:apptable_case1b}
\end{table}

\emph{Case 2: Agents $2$ and $3$  prefer distinct objects as their top choices.}

There are eight profiles that satisfy this condition. We show four of these in Table \ref{tab:apptable_profiles_case2} where agent $2$ has object $b$ as her most preferred while agent $3$ prefers object $c$. The remaining four are symmetrical to these with agent $2$ and agent $3$'s preferences swapped. The argument for these profiles is analogous to that for the four cases we consider below.

\begin{table}[h]
    \begin{subtable}[h]{0.22\textwidth}
        \centering
        \begin{tabular}{llll}
        \\
        \toprule
1 & a & b & c \\
2 & b & a & c \\
3 & c & a & b \\
        \bottomrule
       \end{tabular}
       \caption{Profile I}
       \label{tab:appprofile_2i}
    \end{subtable}
    \begin{subtable}[h]{0.22\textwidth}
        \centering
        \begin{tabular}{llll}
        \\
        \toprule
1 & a & b & c \\
2 & b & a & c \\
3 & c & b & a \\
        \bottomrule
       \end{tabular}
       \caption{Profile J}
       \label{tab:appprofile_2j}
    \end{subtable}
    \begin{subtable}[h]{0.22\textwidth}
        \centering
        \begin{tabular}{llll}
        \\
        \toprule
1 & a & b & c \\
2 & b & c & a \\
3 & c & a & b \\
        \bottomrule
       \end{tabular}
       \caption{Profile K}
       \label{tab:appprofile_2k}
    \end{subtable}
    \begin{subtable}[h]{0.22\textwidth}
        \centering
        \begin{tabular}{llll}
        \\
        \toprule
1 & a & b & c \\
2 & b & c & a \\
3 & c & b & a \\
        \bottomrule
       \end{tabular}
       \caption{Profile L}
       \label{tab:appprofile_2l}
    \end{subtable}\\
    \caption{Four Preference Profiles In Case (2)}
     \label{tab:apptable_profiles_case2}
\end{table}

\begin{description}
\item[Profile I ($\succ^{(I)}$):]
Notice that Profile I is a contention-free preference profile. Since $P^{(I)}_{1,a}=1$, for agents $2$ and $3$ to not envy each other, $P^{(B)}_{2,b}=P^{(C)}_{3,c}=1$. Thus we get a contention-free efficient assignment, which contradicts with the Theorem \ref{thm:main_hardness}.

The arguments for Profiles J and K are similar to Profile I.

\item[Profile L ($\succ^{(L)}$):]
The assignment for Profile L is shown in Table \ref{tab:apptable_profile_i}. If agent $2$ chooses to instead report $c \succ_2 b \succ_2 a$, strategy-proofness and envy-freeness imply that we arrive at the profile and its corresponding assignment in Case 1b for which we have already established a contradiction.

\begin{table}[h]
    \begin{subtable}[h]{0.45\textwidth}
        \centering
        \begin{tabular}{llll}
\\
\toprule
1 & a & b & c \\
2 & b & c & a \\
3 & c & b & a \\
        \bottomrule
       \end{tabular}
    \end{subtable}
    \begin{subtable}[h]{0.45\textwidth}
        \centering
        \begin{tabular}{llll}
 & a & b & c \\
\bottomrule
1 & $1$ & $0$ & $0$ \\
2 & $0$ & $x$ & $1-x$ \\
3 & $0$ & $1-x$ & $x$ \\
        \bottomrule
       \end{tabular}
    \end{subtable}
     \caption{Profile I}
     \label{tab:apptable_profile_i}
\end{table}

\end{description}

\section{Proof of Proposition \ref{prop:matrix_representation}}
\label{prop:matrix_representation_proof}

In order to prove this proposition, we need to first define some graph preliminaries. Let $G=(N,O,E)$ be a bipartite directed graph. The set of agents, $N$, and the set of objects, $O$, are nodes in $G$. For every agent-object pair, we add the restriction that there can exist at most one directed edge between the agent and the object. That is, if $(u,v) \in E$, then $(v,u) \notin E$. A \emph{cycle} in G is a sequence of directed edges, $(u_1,v_1), (u_2,v_2),\ldots (u_n,v_n)$ such that $v_i=u_{i+1}$, $\forall 1 \leq i \leq (n-1)$ and $v_n=u_1$. Let $\mathbb{C}$ be the set of all cycles. A \emph{flow} on the graph $G$ is a function $x:E \rightarrow \bbR_+$ such that the total flow entering a node must equal the flow exiting that node. Formally, $\forall u \in \{N \cup O\}, \sum_{v| (u,v) \in E} x(u,v) = \sum_{v| (v,u) \in E} x(v,u)$. We are now ready to prove the proposition.

Fix a preference profile $\bsucc$ and let $P^{\bsucc}=\varphi(\bsucc)$. Let $E$ be the equal division random assignment matrix. We have $E_{i,a}=\frac{1}{n}$, $\forall i\in N \text{ and } a \in O$. Let $D = P - E$. Notice that $D$ is a zero-sum matrix, i.e., a matrix in which every row and column sums to zero. Also, since $0 \leq P_{i,a} \leq 1$, $\frac{-1}{n} \leq D_{i,a} \leq \frac{n-1}{n}$. We now construct a graph $G$ in the following way: For every $i \in N$ and $a \in O$, if $D_{i,a} > 0$, then we add a directed edge from object node $a$ to agent node $i$. If instead $D_{i,a} < 0$, then we add a directed edge from node $i$ to node $a$. Formally, if $D_{i,a}>0$, then $(a,i) \in E$. If $D_{i,a}<0$, then $(i,a) \in E$. 

Let us define a function $x:E \rightarrow \bbR_+$ as follows: For all $i \in N$ and $a \in O$, if $(a,i) \in E$, then let $x(a,i)=D_{i,a}$. If $(i,a) \in E$, then let $x(i,a)=-D_{i,a}$. The function $x$ clearly represents a valid flow on the graph $G$. By the Flow Decomposition Theorem~\citep{ahuja1988network}, any such flow can be decomposed into a sum of flows around directed cycles. That is for all $(u,v) \in E$,
\begin{align*}
    x(u,v)&=\sum_{\{C \in \mathbb{C} | (u,v) \in C\}}g(C)
\end{align*}
where $g:\mathbb{C} \rightarrow \bbR_+$. 

We can now define the function $h^{\bsucc}$. For all $i,j \in N$ and $a \in O$,
\begin{align*}
    h^{\bsucc}(i,j,a)&=\begin{cases}
    \sum_{\{C \in \mathbb{C} | (j,a) \in C \text{ and } (a,i) \in C\}} g(C) \text{  } \quad (a,i) \in E \\
    -\sum_{\{C \in \mathbb{C} | (i,a) \in C \text{ and } (a,j) \in C\}} g(C) \quad (i,a) \in E \\
    0 \quad \quad \quad \quad \quad \quad \quad \quad \quad \quad \quad \quad \quad \text{ otherwise}
    \end{cases}
\end{align*}
First notice that, whenever $(i,a) \in E$ for some $i \in N$ and $a \in O$, we have $0 \leq x(i,a) \leq \frac{1}{n}$. Since $g(C) \geq 0$, $\forall C \in \mathbb{C}$, $h^{\bsucc}(i,j,a) \geq 0$ whenever $(a,i) \in E$. Additionally, \begin{align*}
    h^{\bsucc}(i,j,a) = \sum_{\{C \in \mathbb{C} | (j,a) \in C \text{ and } (a,i) \in C\}} g(C) \leq \sum_{\{C \in \mathbb{C} | (j,a) \in C\}} g(C) = x(j,a) \leq \frac{1}{n}
\end{align*}
On the other hand, when $(i,a) \in E$, 
\begin{align*}
    h^{\bsucc}(i,j,a) = -\Big[\sum_{\{C \in \mathbb{C} | (i,a) \in C \text{ and } (a,j) \in C\}} g(C)\Big] \geq -\Big[\sum_{\{C \in \mathbb{C} | (i,a) \in C\}} g(C) \Big] = -x(i,a) \geq \frac{-1}{n}
\end{align*}
Combining the inequalities for the two cases above, we have that $h^{\bsucc}(i,j,a) \in [\frac{-1}{n},\frac{1}{n}]$. We also notice that for all $(a,i) \in E$ such that $i \in N$ and $a \in O$, 
\begin{align*}
    x(a,i)=\sum_{\{C \in \mathbb{C} | (a,i) \in C\}}g(C) = \sum_{j \in N \setminus \{i\}} \sum_{\{C \in \mathbb{C} | (j,a) \in C \text{ and } (a,i) \in C\}}g(C) = \sum_{j \in N \setminus \{i\}}h^{\bsucc}(i,j,a)
\end{align*}
Similarly, for all $(i,a) \in E$ such that $i \in N$ and $a \in O$, 
\begin{align*}
    x(i,a)=\sum_{\{C \in \mathbb{C} | (i,a) \in C\}}g(C) = \sum_{j \in N \setminus \{i\}} \sum_{\{C \in \mathbb{C} | (i,a) \in C \text{ and } (a,j) \in C\}}g(C) = \sum_{j \in N \setminus \{i\}}-h^{\bsucc}(i,j,a)
\end{align*}
Using this definition, we can represent the random assignment matrix $P^{\bsucc}$ as,
\begin{align*}
    P^{\bsucc}_{i,a} &= E_{i,a} + D_{i,a}, \quad \quad  \quad \forall  i \in N \text{ and } \forall a \in O \\
    &= \frac{1}{n} + \sum_{j \in N \setminus \{i\}} h^{\bsucc}(i,j,a) \\
\end{align*}
which is our desired representation.

\section{Proof for Proposition \ref{prop:necessary_properties}}
\label{sec:function-properties}

In this section, we prove that the if the mechanism $\varphi^f$ induced by the transfer function is feasible, then $f$ must satisfy the following properties. For convenience, let $P^{\bsucc}=\varphi^f(\bsucc)$ be the random assignment associated with a profile $\bsucc$. 

\begin{description}
\item[\ref{item:identical}] \emph{No transfer:} $f(\succ,\succ,a)=0$, $\forall \succ \in \calR$ and $\forall a \in O$.
    
Suppose for contradiction that there exists $\succ \in \calR$ and $a \in O$ such that $f(\succ,\succ,a) \neq 0$. Consider a preference profile $\bsucc = (\succ_i)_{i \in N}$ such that $\succ_i=\succ$ for $i \in N$. However, by definition

\begin{align*}
    \sum_{i \in N}P^{\bsucc}_{i,a} &= \sum_{i \in N}\frac{1}{n} + \sum_{i \in N}\sum_{j \neq i}f(\succ_i,\succ_j,a) \\
    &= 1 + n \cdot (n-1)\cdot f(\succ,\succ,a) \\
    &\neq 1
\end{align*}
This contradicts with the feasibility of the random assignment $P^{\bsucc}$.

\item[\ref{item:zero}] \emph{Balanced transfers:} $\sum_{a \in O} f(\succ, \succ', a) = 0$ for any pair of preferences $\succ, \succ' \in \calR$.

We prove this by contradiction. Suppose that there exist a pair of preferences $\succ, \succ' \in \calR$ such that $\sum_{a \in O} f(\succ,\succ',a) \neq 0$. Consider a preference profile $\bsucc=(\succ_i,\bsucc_{-i})$ such that $\succ_i=\succ$ and $\succ_j = \succ'$ for $j \neq i$. The total allocation for agent $i$ is given by

\begin{align*}
    \sum_{a \in O}P^{\bsucc}_{i,a} &= \sum_{a \in O}\frac{1}{n} + \sum_{a \in O}\sum_{j \neq i}f(\succ_i,\succ_j,a) \\
    &= 1 + \sum_{a \in O}\sum_{j \neq i}f(\succ,\succ',a)\\
    &= 1+ (n-1) \cdot \sum_{a \in O} f(\succ,\succ',a) \\
    &\neq 1
\end{align*}
which contradicts feasibility of $P^{\bsucc}$.

\item[\ref{item:anti-symmetry}] \emph{Anti-symmetry:} $f(\succ, \succ', a) = - f(\succ', \succ, a)$ for any pair of preferences $\succ, \succ' \in \calR$ and $\forall a \in O$. 
We again proceed by contradiction. Suppose $\succ,\succ' \in \calR$ and $a \in O$ are such that $f(\succ,\succ',a)\neq -f(\succ',\succ,a)$. Let us again consider the preference profile $\bsucc=(\succ_k,\bsucc_{-k})$ such that $\succ_k=\succ$ and $\succ_j = \succ'$ for $j \neq k$. If we sum the probability of being assigned object $a$ over all agents
\begin{align*}
    \sum_{i \in N}P^{\bsucc}_{i,a} &= \sum_{i \in N}\frac{1}{n} + \sum_{i \in N}\sum_{j \neq i}f(\succ_i,\succ_j,a) \\
    &= 1 + \sum_{j \neq k}f(\succ_k,\succ_j,a) + \sum_{i \neq k}\sum_{j \neq i}f(\succ_i,\succ_j,a) \\
    &= 1 + \sum_{j \neq k}f(\succ,\succ',a) + \sum_{i \neq k}\sum_{j \in N\setminus\{i,k\}}f(\succ_i,\succ_j,a) + \sum_{i \neq k}f(\succ_i,\succ_k,a)  \\
    &= 1 + \sum_{j \neq k}f(\succ,\succ',a) + \sum_{i \neq k}\sum_{j \in N\setminus\{i,k\}}f(\succ',\succ',a)+\sum_{i \neq k}f(\succ',\succ,a)
    \intertext{Since $f$ satisfies the property that there are no transfers between identical preferences (Property (\ref{item:identical})) }
    \sum_{i \in N}P^{\bsucc}_{i,a} &= 1 + \sum_{j \neq k}f(\succ,\succ',a) + \sum_{i \neq k}f(\succ',\succ,a) \\
    &\neq 1
\end{align*}
which leads to a contradiction

\item[\ref{item:sum}] \emph{Bounded range:} $f(\succ, \succ', a) \in [-\frac{1}{n(n-1)}, \frac{1}{n(n-1)}]$ for any $\succ, \succ' \in \calR$ and $\forall a \in O$.

Consider a preference profile $\bsucc = (\succ_i, \bsucc_{-i})$ such that $\succ_i = \succ$ and $\succ_j = \succ'$ for all $j \neq i$. Since $P^{\bsucc}_{i,a} \geq 0$, we have
\begin{align*}
    \frac{1}{n} + \sum_{j \neq i} f(\succ_i, \succ_j, a) &\geq 0\\
    \frac{1}{n} + (n-1) f(\succ, \succ', a) &\geq 0 \implies f(\succ, \succ', a) \geq \frac{-1}{n(n-1)}
\end{align*}
Exchanging the roles of $\succ$ and $\succ'$, we similarly obtain $f(\succ', \succ, a) \geq \frac{-1}{n(n-1)}$.
Thus, by  the \emph{anti-symmetry} property defined above, we have $f(\succ, \succ', a) = -f(\succ', \succ, a) \leq \frac{1}{n(n-1)}$.
\end{description}

\section{Proofs from Section \ref{subsec:equiv}}
\label{app:proof_sp_equiv_ef}

\spmonotone*
\begin{proof}
We prove the contrapositive. Suppose there exists $\succ',\succ'' \in \calR$ and some $t \in [n]$, such that $\sum_{k=1}^t f(\succ,\succ'',a_k) < \sum_{k=1}^t f(\succ',\succ'',a_k)$. Consider $\bsucc$ such that $\succ_i = \succ$ for some $i \in N$ and $\succ_j=\succ''$ for all $j \neq i$. Let $\bsucc'$ be a preference profile where agent $i$ lies and reports $\succ'_i=\succ'$ while all other agents continue to report $\succ''$ as in $\bsucc$. Let $P^{\bsucc}=\varphi^f(\bsucc)$ and $P^{\bsucc'}=\varphi^f(\bsucc')$. We have that,
\begin{align*}
\sum_{k=1}^tP^{\bsucc}_{i,a_k} &= \frac{t}{n} + \sum_{j \neq i} \sum_{k=1}^t f(\succ_i,\succ_j,a_k) =\frac{t}{n} + \sum_{j \neq i} \sum_{k=1}^t f(\succ,\succ'',a_k) \\
&=\frac{t}{n} + (n-1) \cdot \sum_{k=1}^t f(\succ,\succ'',a_k) \\
&< \frac{t}{n} + (n-1) \cdot \sum_{k=1}^t f(\succ',\succ'',a_k) = \frac{t}{n} + \sum_{j \neq i} \sum_{k=1}^t f(\succ'_i,\succ_j,a_k) \\
&= \sum_{k=1}^tP^{\bsucc'}_{i,a_k}
\end{align*}
Thus $\varphi^f$ is not strategy-proof.
\end{proof}

\fupperlowerswap*
\begin{proof}
That swap monotonicity and sender invariance of the transfer function $f$ implies strategy-proofness of $\varphi^f$ follows from Lemma \ref{lemma:sp_swap_upper_lower} and definitions.

For the necessity part of the lemma, we first show that strategy-proofness implies sender invariance. Let $\succ' = \langle a_1, a_2, \ldots, a_n\rangle$, $\succ''$ be such that $\sigma(\succ', k)=\sigma(\succ'',k)$, $\forall k \leq t$ and consider any $\succ \in \calR$. At a preference profile $\bsucc$, where $\succ_i=\succ'$ for some agent $i \in N$ and $\succ_j=\succ$ for all other agents $j \in N \setminus \{i\}$, a mis-report $\succ'_i=\succ''$ for agent $i$, strategy-proofness implies that agent $i$ must receive the same allocation for her top $(t-1)$ objects as well her top $t$ objects under $\succ_i$ and $\succ'_i$ since $\sigma(\succ', k)=\sigma(\succ'',k)$, $\forall k \leq t$.  Let $P^{\bsucc} = \varphi^f(\bsucc)$ and $P^{(\succ'_i,\bsucc_{-i})}=\varphi^f(\succ'_i,\bsucc_{-i})$. We have,
\begin{align}
    \sum_{k=1}^{t-1} P^{\bsucc}_{i,a_k} &= \sum_{k=1}^{t-1} P^{(\succ'_i,\bsucc_{-i})}_{i,a_k} \\
    \implies \frac{t-1}{n} + (n-1)\sum_{k=1}^{t-1} f(\succ',\succ,a_k) &= \frac{t-1}{n} + (n-1)\sum_{k=1}^{t-1} f(\succ'',\succ,a_k) \label{eqn:sender_sp_t_1}
    \intertext{and}
    \sum_{k=1}^{t} P^{\bsucc}_{i,a_k} &= \sum_{k=1}^{t} P^{(\succ'_i,\bsucc_{-i})}_{i,a_k} \\
    \implies \frac{t}{n} + (n-1)\sum_{k=1}^{t} f(\succ',\succ,a_k) &= \frac{t}{n} + (n-1)\sum_{k=1}^{t} f(\succ'',\succ,a_k) \label{eqn:sender_sp_t}
\end{align}
Equations \eqref{eqn:sender_sp_t_1} and \eqref{eqn:sender_sp_t} together imply that $f(\succ',\succ,a_t) = f(\succ'',\succ,a_t)$.
Since $f$ satisfies anti-symmetry, we have our desired result that $f(\succ,\succ',a_t) = f(\succ,\succ'',a_t)$.

Next, if the transfer function $f$ is not swap-monotonic, then there exists $\succ' \in \calR$, which we can without loss of generality assume to be $\succ'=\langle a_1, a_2, \ldots, a_n \rangle$, $\succ'' \in \Gamma(\succ')$ such that $a_t \succ' a_{t+1}$ but $a_{t+1} \succ'' a_t$ for some $t \in [n]$ and $\succ \in \calR$ such that $f(\succ',\succ,a_{t}) < f(\succ'',\succ,a_{t})$. At the preference profile $\bsucc$, where $\succ_i=\succ'$ for some agent $i$ and $\succ_j=\succ$ for all other agents $j \in N \setminus \{i\}$, if agent $i$ chooses to deviate and report $\succ'_i=\succ''$, we have
\begin{align*}
    P^{\bsucc}_{i,a_t} = \frac{1}{n} + (n-1)f(\succ',\succ,a_t) < \frac{1}{n} + (n-1)f(\succ'',\succ,a_t) = P^{(\succ'_i,\bsucc_{-i})}_{i,a_t}
\end{align*}
which implies that the mechanism $\varphi^f$ is not swap-monotonic and hence is not strategy-proof.
\end{proof}

\efzerotransfers*
\begin{proof}
Let $\succ=\langle a_1,a_2,\ldots,a_n\rangle$ and $\succ'$ be such that $\sigma(\succ, k)=\sigma(\succ',k)$ $\forall k \leq t$. Consider a preference profile $\bsucc$ where $\succ_i=\succ$ for some agent $i$, $\succ_{j}=\succ'$ for all other agents $j \in N \setminus \{i\}$. For agent $i$ to not envy any other agent $j \in N \setminus \{i\}$, they must get the same allocation for agent $i$'s top $t-1$ and top $t$ objects, since $\sigma(\succ_i, k)=\sigma(\succ_j,k)$ $\forall k \leq t$, . That is,
\begin{align*}
    \sum_{k=1}^{t-1} P^{\bsucc}_{i,a_k} &= \sum_{k=1}^{t-1} P^{\bsucc}_{j,a_k} \text{ where $j \in N \setminus \{i\}$} \\
    \implies \frac{t-1}{n} + (n-1)\sum_{k=1}^{t-1} f(\succ,\succ',a_k) &= \frac{t-1}{n} + \sum_{k=1}^{t-1} f(\succ',\succ,a_k) +  (n-2)\sum_{k=1}^{t-1} f(\succ',\succ',a_k)
    \intertext{By the no transfer property from Proposition \ref{prop:necessary_properties}, $f(\succ',\succ',o)=0$ for all $o \in O$, which in conjunction with the anti-symmetry property of $f$ implies that}
    \sum_{k=1}^{t-1} f(\succ,\succ',a_k) &= 0
    \intertext{Similarly we get, }
    \sum_{k=1}^{t} f(\succ,\succ',a_k) &= 0   \\
    \implies f(\succ,\succ',a_t)&=0
\end{align*}
The argument for the second part of the lemma is symmetric to the one made above.
\end{proof}

\eftransferatleastzero*
\begin{proof}
    Consider a preference profile $\bsucc$ such that $\succ_i=\succ$ for some agent $i$ and $\succ_j=\succ'$ for all agents $j \in N \setminus \{i\}$. By envy-freeness, the allocation that agent $i$ receives for her top $t$ objects is at least as much as the allocation any other agent $j$ receives. That is,
    \begin{align*}
        \sum_{k=1}^t P^{\bsucc}_{i,a_k} &\geq \sum_{k=1}^t P^{\bsucc}_{j,a_k} \\
        \implies \frac{t}{n} + (n-1) \cdot \sum_{k=1}^t f(\succ,\succ',a_k) &\geq \frac{t}{n} + \sum_{k=1}^t f(\succ',\succ,a_k) + (n-2) \cdot \sum_{k=1}^t f(\succ',\succ',a_k) \\
        \implies \sum_{k=1}^t f(\succ,\succ',a_k) &\geq 0
    \end{align*}
where the last inequality follows from the no transfer and anti-symmetry properties.
\end{proof}

\section{Proof of Lemma \ref{lem:function-properties-necessary}}
\label{app:lemma-necessary}
From the definition of pairwise exchange mechanisms, it follows that mechanism $\varphi_f$ is neutral if and only if the function $f$ is neutral, i.e., it satisfies property \eqref{item:neutral}.

We now proceed to show that the receiver invariance property \eqref{item:receiver-invariance} is necessary. Suppose for contradiction that we have $\succ = \langle a_1, a_2, \ldots, a_n\rangle$ and $\succ'$ and $\succ''$ such that $rank(\succ', a_k)=rank(\succ'',a_k)$ $\forall k\leq t$, but $f(\succ,\succ',a_t)\neq f(\succ,\succ'',a_t)$. Without loss of generality, let $f(\succ,\succ',a_t) < f(\succ,\succ'',a_t)$. Fix an agent $i \in N$. Now consider a preference profile $\bsucc=(\succ_i,\bsucc_{-i})$ where $\succ_i = \succ$ and $\succ_j = \succ', \forall j \neq i$. 
        Let $\pi: O\rightarrow O$ be a permutation such that $\pi(\succ'')=\succ'$. Observe that since $rank(\succ', a_k)=rank(\succ'',a_k)$ $\forall k\leq t$, we have $\pi(a_k)=a_k$ $\forall k\leq t$. Let $\succ'_i=\pi(\succ_i)$ be the preference that results from re-labeling objects in $\succ_i$ using $\pi$. Note that in both $\succ_i$ and $\succ'_i$, we will have $\sigma(\succ_i,k)=\sigma(\succ'_i,k)$ $\forall k \leq t$. Since $f$ satisfies neutrality, $f(\succ_i,\succ'',a_t)=f(\pi(\succ_i),\pi(\succ''),\pi(a_t))=f(\succ'_i,\succ',a_t)$, which in conjunction with our assumption $f(\succ,\succ'',a_t)>f(\succ,\succ',a_t)=f(\succ_i,\succ',a_t)$ implies that $f(\succ'_i,\succ',a_t) > f(\succ_i,\succ',a_t)$.  
    Now, starting from $\bsucc$, suppose agent $i$ chooses to instead report her preference as $\succ'_i$. Let $\bsucc'=(\succ'_i,\bsucc_{-i})$ be the resulting profile. Since the mechanism is strategy-proof, $\sum_{k=1}^{t-1}P^{\bsucc}_{i,a_k}=\sum_{k=1}^{t-1}P^{\bsucc'}_{i,a_k}$. The probability that agent $i$ is assigned the object $a_t$ when she reports her true preferences is
    \begin{align*}
        P^{\bsucc}_{i,a_t} &= \frac{1}{n} + \sum_{j \neq i} f(\succ_i,\succ_j,a_t) =  \frac{1}{n} + (n-1) \cdot f(\succ_i,\succ',a_t) \\
        \intertext{The second equality above comes from the construction of our profile, where $\succ_j=\succ'$ $\forall j \neq i$. If agent $i$ reports $\succ'_i$ we have}
        P^{\bsucc'}_{i,a_t} &= \frac{1}{n} + \sum_{j \neq i} f(\succ'_i,\succ_j,a_t) \\
        &= \frac{1}{n} + (n-1) \cdot f(\succ'_i,\succ',a_t) \\
        &> \frac{1}{n} + (n-1) \cdot f(\succ_i,\succ',a_t) \\
        &= P^{\bsucc}_{i,a_t}
    \end{align*}
    which leads to a contradiction. The proof for the second part of the receiver invariance property follows symmetrically.

\section{Proof of Base Case of Theorem \ref{thm:main_rank}}
\label{app:theorem-base-case}

For the base case, let $n=3$ and $O=\{a_1,a_2,a_3\}$. Since $f$ is neutral, we will let $\succ=\langle a_1,a_2,a_3 \rangle$. There are six preferences in $\calR$. We list these preferences and the transfer functions for $\succ$ with each of them in Table \ref{tab:basecase}.

By the receiver invariance property, we know that for object $a_1$ and $\forall \succ',\succ'' \in \calR$
\begin{align*}
    rank(\succ',a_1)=rank(\succ'',a_1) \implies f(\succ,\succ',a_1)=f(\succ,\succ'',a_1)
\end{align*}
Let $x_i=f(\succ,\succ',a_1)$ be the transfers when $rank(\succ',a_1)=i$. By the same property, we also have that for object $a_3$ and $\forall \succ',\succ'' \in \calR$
\begin{align*}
    rank(\succ',a_3)=rank(\succ'',a_3) \implies f(\succ,\succ',a_3)=f(\succ,\succ'',a_3)
\end{align*}
Let $y_i=f(\succ,\succ',a_3)$ be the transfers when $rank(\succ',a_3)=i$.

By the sender invariance property for object $a_2$, if $rank(\succ',a_2)=rank(\succ'',a_2)=1$ or $rank(\succ',a_2)=rank(\succ'',a_2)=3$, then $f(\succ,\succ',a_2)=f(\succ,\succ'',a_2)$. Let $w_i=f(\succ,\succ',a_2)$ when $rank(\succ',a_2)=i$ for $i \in \{1,3\}$. It remains to be shown that the claim is true when $rank(\succ',a_2)=rank(\succ'',a_2)=2$.

Since $\sum_{a \in O} f(\succ, \succ', a) = c$, we have
\begin{align*}
    \sum_{a \in O} f(\succ, \langle a_1,a_3,a_2 \rangle, a) + \sum_{a \in O} f(\succ, \langle a_2,a_1,a_3 \rangle, a) &= \sum_{a \in O} f(\succ, \langle a_2,a_3,a_1 \rangle, a) + \sum_{a \in O} f(\succ, \langle a_3,a_1,a_2 \rangle, a) \\
    x_1 + w_3 + y_2 + x_2 + w_1 + y_3 &= x_3 + w_1 + y_2 + x_2 + w_3 + y_1 \\
    x_1 + y_3 &= x_3 + y_1
\end{align*}
This implies that when $rank(\succ',a_2)=rank(\succ'',a_2)=2$, $f(\succ,\succ',a_2)=f(\succ,\succ'',a_2)$ as desired.

\begin{table}[h]
        \centering
        \begin{tabular}{ccccccc}
& $\langle a_1,a_2,a_3\rangle$ & $\langle a_1,a_3,a_2\rangle$ & $\langle a_2,a_1,a_3\rangle$ & $\langle a_2,a_3,a_1\rangle$ & $\langle a_3,a_1,a_2\rangle$ & $\langle a_3,a_2,a_1\rangle$ \\
\bottomrule
$a_1$ & $x_1$ & $x_1$ & $x_2$ & $x_3$ & $x_2$ & $x_3$\\
$a_2$ & & $w_3$ & $w_1$ & $w_1$ & $w_3$ & \\
$a_3$ & $y_3$ & $y_2$ & $y_3$ & $y_2$ & $y_1$ & $y_1$\\
\bottomrule
       \end{tabular}
       \caption{Transfer Function $f$ for $n=3$}
       \label{tab:basecase}
    
\end{table}

\section{Proof of Proposition \ref{prop:rankfunction}}
\label{app:proposition-proof}

Let's prove the easy direction first. If there exists a function $g$ such that $f(\succ, \succ', a) = g(rank(\succ, a), rank(\succ',a))$, then for any two preferences $\succ, \succ'$ such that $rank(\succ', a) = rank(\succ'', a) = r$, we have 
$f(\succ, \succ', a) = g(rank(\succ, a), r) = f(\succ, \succ'', a)
$ as desired.

For the other direction, let $\succ^*$ be an arbitrary preference relation and let $\succ^* = \langle a_1, a_2, \ldots, a_n\rangle$ without loss of generality. We define a function $g : [n] \times [n] \rightarrow [-1,1]$ as follows - 
\begin{align*}
    g(r,s) = f(\succ^*, \hat{\succ}, a_r) \text{ where $\hat{\succ}$ is an arbitrary preference such that $rank(\hat{\succ}, a_r) = s$.}
\end{align*}
Note that this function is well defined since we have $f(\succ^*, \hat{\succ}, a_r) = f(\succ^*, \hat{\succ}', a_r)$ for any other preference relation $\hat{\succ}'$ that has $rank(\hat{\succ}', a_r) = s$. Now consider any two preferences $\succ$ and $\succ'$ and an object $a \in O$. Suppose $rank(\succ,a) = \ell$. Let $\pi$ be a permutation such that $\succ^* = \pi(\succ)$, and define $\tilde{\succ} = \pi(\succ')$. By definition, we have $rank(\tilde{\succ}, a_\ell) = rank(\succ', a)$.
Since the function $f$ is neutral, we must have
\begin{align*}
    f(\succ, \succ', a) &= f(\succ^*, \tilde{\succ}, a_\ell) = g(\ell, rank(\tilde{\succ}, a_\ell)) = g(rank(\succ, a), rank(\succ', a)) \qedhere
\end{align*}

\section{Proof of Theorem \ref{thm:main-charac}}
\label{app:thm-pairwise-charac}

For the first direction, we observe that Lemma \ref{lem:function-properties-necessary}, Theorem \ref{thm:main_rank}, Proposition \ref{prop:rankfunction}, and Theorem \ref{thm:g-to-linear} together imply the existence of a vector $v \in \bbR^n$ such that $f(\succ, \succ', a) = v[rank(\succ,a)] - v[rank(\succ',a)]$. The additional restrictions on the vector $v$ stem from property \eqref{item:sum} in Proposition \ref{prop:necessary_properties} and the strategy-proofness of the mechanism $\varphi^f$. If the vector $v$ is not sorted, it is easy to find a profile that violates the strategy-proofness of the mechanism. The restriction of $v \in [0,\frac{1}{n(n-1)}]^n$ with $v[0]$ is needed to ensure that $f(\succ, \succ', a) \in [-\frac{1}{n(n-1)},\frac{1}{n(n-1)}]$.

We now show that a vector $v$ with the above-mentioned properties is sufficient to guarantee that the mechanism $\varphi_f$ is neutral, strategy-proof, and envy-free. First, to see that the mechanism $\varphi_f$ is a feasible random assignment mechanism, we observe that in any profile $\bsucc$, we have $P^{\bsucc}_{i,a} = \dfrac{1}{n} + \sum_{j \neq i} (v[rank(\succ_i, a)] - v[rank(\succ_j,a)])$. But since for any $m, k$, we have $v[m] - v[k] \in [\dfrac{-1}{n(n-1)}, \dfrac{1}{n(n-1)}]$, this implies that $P^{\bsucc}_{i,a} \in [0, 1]$ as desired. The total allocation over all objects for any agent $i$ is given by
\begin{align*}
    \sum_{a \in O} P^{\bsucc}_{i,a} &= \sum_{a \in O} \dfrac{1}{n} + \sum_{a \in O} \sum_{j \neq i}(v[rank(\succ_i,a)]-v[rank(\succ_j, a)]) \\
    &= 1 + \sum_{j \neq i} (\sum_{a \in O} v[rank(\succ_i,a)] - \sum_{a \in O}v[rank(\succ_j, a)])
    = 1
\end{align*}

Similarly, the total allocation of an object $a \in O$ over all agents is 
\begin{align*}
    \sum_{i \in N} P^{\bsucc}_{i,a} &= \sum_{i \in N} \dfrac{1}{n} + \sum_{i \in N} \sum_{j \neq i}(v[rank(\succ_i,a)]-v[rank(\succ_j, a)]) \\
    &= 1 + \sum_{i \in N}((n-1)\cdot v[rank(\succ_i,a)] - \sum_{j \neq i}v[rank(\succ_j, a)]) \\
    &= 1 + (n-1)\cdot \sum_{i \in N} v[rank(\succ_i,a)] - \sum_{i \in N}\sum_{j \neq i}v[rank(\succ_j, a)] \\
    &= 1 + (n-1)\cdot \sum_{i \in N} v[rank(\succ_i,a)] - (n-1) \cdot \sum_{i \in N} v[rank(\succ_i, a)] \\
    &= 1
\end{align*}

Thus, the assignment obtained is doubly stochastic and hence the mechanism is feasible. The mechanism is neutral by definition since the function $f$ is neutral. Finally, we show that the mechanism $\varphi_f$ is also strategy-proof and envy-free.

Consider two agents $i, i' \in N$ and any preference profile $\bsucc$. Let without loss of generality that $\succ_i = \langle a_1, \ldots, a_n \rangle$. For any $t \in [n]$, let us consider the total allocation obtained by agents $i$ and $i'$ for the top $t$ objects in agent $i$'s preference. We have the following.
\begin{align}
    \sum_{k \leq t} P_{i,a_k} &= \dfrac{t}{n} + \sum_{j \in N \setminus \{i\}} \sum_{k \leq t} (v[k] - v[rank(\succ_j, a_k)])\\
    &= \dfrac{t}{n} + \sum_{j \in N \setminus \{i,i'\}} \sum_{k \leq t} (v[k] - v[rank(\succ_j, a_k)]) + \sum_{k \leq t} (v[k] - v[rank(\succ_{i'}, a_k)]) 
    \intertext{However, for any agent $i'$, $\sum_{k \leq t} v[k] \geq \sum_{k \leq t} v[rank(\succ_{i'}, a_k)]$ since the vector $v$ is sorted, and hence we have}
    &\geq \dfrac{t}{n} + \sum_{j \in N \setminus \{i,i'\}} \sum_{k \leq t} (v[rank(\succ_{i'}, a_k)] - v[rank(\succ_j, a_k)]) + \sum_{k \leq t} (v[rank(\succ_{i'}, a_k)] - v[k])\\
    &= \sum_{k \leq t} P_{i',a_k}
\end{align}
and hence the mechanism $\phi_f$ is envy-free.

To show strategy-proofness, consider a profile $\bsucc=(\succ_i, \bsucc_{-i})$ and another profile $\bsucc' = (\succ_i', \bsucc_{-i})$ where agent $i$ misreports her preferences. Let $P = \varphi_f(\bsucc)$ and $P' = \varphi_f(\bsucc')$ denote the corresponding random assignments obtained by the mechanism. We have the following.
\begin{align}
    \sum_{k \leq t} P_{i,a_k} &= \dfrac{t}{n} + \sum_{j \in N \setminus \{i\}} \sum_{k \leq t} (v[k] - v[rank(\succ_j, a_k)])
    \intertext{Again, since the vector $v$ is sorted, we have $\sum_{k \leq t} v[k] \geq \sum_{k \leq t} v[rank(\succ'_i, a_k)]$ for any $\succ'_i$, and hence we have}
    &\geq \dfrac{t}{n} + \sum_{j \in N \setminus \{i\}} \sum_{k \leq t} (v[rank(\succ'_{i}, a_k)] - v[rank(\succ_j, a_k)]) = \sum_{k \leq t} P'_{i,a_k}
\end{align}
and thus truth-telling is a dominant strategy for any agent $i$ and the mechanism is strategy-proof.

\section{Proof of Theorem \ref{thm:linear-pareto-efficient}}
\label{app:linear-pareto-efficient}

We first show that the condition above is sufficient to guarantee that $\varphi^v$ is pareto-efficient within the class $\cF$. Suppose for contradiction that there exists a vector $u \in [0, \dfrac{1}{n(n-1)}]^n$ where $u \neq v$ such that the mechanism $\varphi^u$ dominates $\varphi^v$. 

We first claim that any such vector $u$ must satisfy $u[k] \leq v[k],\ \forall k \in [n]$. Indeed, since $u[1] \leq \dfrac{1}{n(n-1)} = v[1]$, the claim is trivially true for $k=1$. For any $k > 1$, consider a profile $\bsucc$ where the agent 1 has the preference $\succ_1 = \langle a_1, a_2, \ldots, a_n\rangle$, while all other agents have object $a_1$ in the $k^{\text{th}}$ position in their preference, i.e. $rank(\succ_j, a_1) = k,\ \forall j \neq i$. Let $Q_{1,1}$ and $P_{1,1}$ denote the allocation received by agent 1 for object $a_1$ in mechanism $\varphi^u$ and $\varphi^v$ respectively.
We have,
\begin{align}
    Q_{1,1} &= \dfrac{1}{n} + \sum_{j \neq i}(u[1] - u[k]) = \dfrac{1}{n} + (n-1)(u[1] - u[k])\\
    &\leq P_{1,1} + (n-1)(v[k]-u[k])
    \intertext{Our assumption that $\varphi^u$ dominates $\varphi^v$ implies that $Q_{1,1} \geq P_{1,1}$, and hence we have}
    v[k] &\geq u[k],\ \forall k \in [n] \label{eq:vlarge}
\end{align}

On the other hand, we can also show that $u[k] \geq v[k],\ \forall k \in [n]$. Since $v[n] = 0$, this is trivially true for $k = n$. For any $k <n$, consider a profile $\bsucc$ where the agent 1 has the preference $\succ_1 = \langle a_1, a_2, \ldots, a_n\rangle$, while all other agents have object $a_n$ in the $k^{\text{th}}$ position in their preference, i.e. $rank(\succ_j, a_n) = k,\ \forall j \neq i$. Let $Q_{1,n}$ and $P_{1,n}$ denote the allocation received by agent 1 for object $a_n$ in mechanism $\varphi^u$ and $\varphi^v$ respectively.
We have,
\begin{align}
    Q_{1,n} &= \dfrac{1}{n} + \sum_{j \neq i}(u[n] - u[k]) = \dfrac{1}{n} + (n-1)(u[n] - u[k])\\
    &\leq P_{1,n} + (n-1)(v[k]-u[k])
    \intertext{Our assumption that $\varphi^u$ dominates $\varphi^v$ implies that $Q_{1,n} \leq P_{1,n}$, and hence we have}
    v[k] &\leq u[k],\ \forall k \in [n] \label{eq:vsmall}
\end{align}

Inequalities \eqref{eq:vlarge} and \eqref{eq:vsmall} together imply that $u[k] = v[k],\ \forall k \in [n]$ which is a contradiction since we assumed that $u \neq v$.

We next prove that the condition $v[1]=\frac{1}{n(n-1)}$ is necessary for Pareto-efficiency within the family of mechanisms $\mathcal{F}$. We prove the contrapositive of the statement. Suppose $v[1] < \frac{1}{n(n-1)}$.  Let $u \in [0, \dfrac{1}{n(n-1)}]^n$ be such that $u[1]=\frac{1}{n(n-1)}$ and $u[k]=v[k]$ for $k \in [n] \setminus \{1\}$. For a profile $\bsucc \in \calR$, let $P^{\bsucc}=\varphi^v(\bsucc)$ and $Q^{\bsucc}=\varphi^u(\bsucc)$. Let without loss of generality that $\succ_i = \langle a_1, \ldots, a_n \rangle$. For any $t \in [n]$, let us consider the total allocation obtained by agent $i$ for her top $t$ objects under $\varphi^u$ and $\varphi^v$.
\begin{align*}
        \sum_{m \leq t} P^{\bsucc}_{i,a_m} &= \dfrac{t}{n} + \sum_{j \in N \setminus \{i\}} \sum_{m \leq t} (v[m] - v[rank(\succ_j, a_m)])\\
        \intertext{Notice that if $rank(\succ_j,a_\ell)=1$ for an agent $j$ and some $\ell \leq t$, then $\sum_{m \leq t} (v[m] -  v[rank(\succ_j,a_m)]) =  \sum_{m \leq t} (u[m] - u[rank(\succ_j,a_m)])$. On the other hand, if $rank(\succ_j,a_m) \neq 1$ for any agent $j$ and all $m \leq t$, then $\sum_{m \leq t} (v[m] - (v[rank(\succ_j,a_m)]) <  \sum_{m \leq t} (u[m] -  u[rank(\succ_j,a_m)])$. Therefore}.
        \sum_{m \leq t} P^{\bsucc}_{i,a_m} & \leq \dfrac{t}{n} + \sum_{j \in N \setminus \{i\}} \sum_{m \leq t} (u[m] - u[rank(\succ_j, a_m)]) =  \sum_{m \leq t} Q^{\bsucc}_{i,a_m}
\end{align*}
Further, the above inequality is strict at any profile $\bsucc$ where there exists an agent $j$ with $\sigma(\succ_j,1) \neq \sigma(\succ_i,1)$ and $t=1$. Therefore, $\varphi^u$ dominates $\varphi^v$.

\end{appendix}

\end{document}